\documentclass[a4wide,12pt]{article}

\usepackage{a4wide}

\usepackage[english]{babel}
\usepackage[T1]{fontenc}
\usepackage{ae,aecompl}
\usepackage{amsfonts}
\usepackage{amssymb}
\usepackage{amsmath}
\usepackage{stmaryrd}
\usepackage{enumitem}
\usepackage{verbatim}
\usepackage{hyperref}

\usepackage{tikz}

\usepackage[font=footnotesize]{caption}

\newtheorem{theorem}{Theorem}[section]
\newtheorem{lemma}[theorem]{Lemma}

\newtheorem{proposition}[theorem]{Proposition}

\newtheorem{definition}[theorem]{Definition}
\newtheorem{remark}[theorem]{Remark}
\newtheorem{example}[theorem]{Example}

\newenvironment{proof}[1][]{\paragraph{Proof{#1}}}{\hfill\null$\square$}

\newcommand{\techrep}[1]{#1}
\newcommand{\lorip}[1]{}
\newcommand{\srp}[1]{ }

\newcommand{\future}[1]{}
\newcommand{\defstyle}{\textbf}

\usepackage{xspace}

\newcommand{\Logicname}[1]{\ensuremath{\mathsf{#1}}}

\newcommand{\WLC}{\Logicname{WLC}\xspace}

\newcommand{\cir}{\ensuremath{\mathsf{{cir}}\xspace}}

\newcommand{\cut}[1]{}

\begin{document}

\title{Rational coordination \\ with no communication or conventions}

\author{
Valentin Goranko$^{1}$,
Antti Kuusisto$^{2,3}$, 
Raine R\"onnholm$^{3,4}$
\\ \\
\normalsize $^{1}$Stockholm University, Sweden \\
\normalsize $^{2}$University of Helsinki, Finland\\
\normalsize $^{3}$Tampere University, Finland\\
\normalsize $^{4}$ENS Paris-Saclay, France
}

\date{}

\maketitle

\begin{abstract}
We study pure coordination games where in every outcome, all players have identical payoffs, `win' or `lose'. We identify and discuss a range of 
`purely rational principles' guiding the reasoning of rational players in such games and compare the  classes of coordination games that can be solved by such players with no preplay communication or conventions.
We observe that it is highly nontrivial to delineate a boundary between purely rational principles and other decision methods, such as conventions, for solving such coordination games.
\end{abstract}

\medskip

\section{Introduction} 
\label{sec:intro} 

\emph{Coordination games} (\cite{Lewis69})   
are games in strategic form with several pure strategy Nash equilibria with the same or comparable payoffs for every player. In these games, all players have the mutual interest to select one of these equilibria. 
In \emph{pure coordination games} (\cite{Lewis69}), aka \emph{games of common payoffs} (\cite{Leyton2008}), all players in the game receive the same payoffs and thus the players have fully aligned preferences to coordinate with each other in order to reach the best possible outcome for everyone. 
In this paper we study one-step \emph{pure win-lose coordination games} (\WLC games) in which all payoffs are either  $1$ (i.e., \emph{win}) or $0$ (i.e., \emph{lose}). 
%

Clearly, if players can communicate when playing a pure coordination game with at least one winning outcome, then they can simply agree on a winning strategy profile, so the game is trivialised. 
What makes such games non-trivial is limited (or non-existing) possibility of preplay communication amongst the players\footnote{Note that, while the common use of `preplay communication' in game theory means communication before the given game is played, here we also mean communication before the players are \emph{even presented with the game}.}, meaning that the players must make their choices based on individual reasoning---without any contact with the other players before (or during) playing the game.

There are many natural real-life situations where such coordination scenarios occur.
We give two examples:

\begin{itemize}
\item[(A)] two cars driving towards each other (on a narrow street) that can avoid a collision by swerving either to the right or left,

\item[(B)] a group of $n$ people who get separated in a city and must individually decide on a place where to get together (`regroup'), supposing the group members do not have any way of contacting each other. 
\end{itemize}

%
The notion of \emph{convention} is an important concept that emerges in the context of coordination.
%
Following Lewis' seminal book \cite{Lewis69}, much of the literature on the topic focuses on \emph{social conventions} which are not explicit agreements established via communication, but rather regularities emerging in social behaviour that every individual believes everyone to follow. If indeed followed by everyone, such conventions help resolve social coordination problems. For example, most words used in natural languages are examples of naturally emerged social conventions, whereas using the metric system instead of some other set of measures is likely to be
mainly based on an explicit agreement.

Both types of conventions may be suitable in different situations. For instance, 
in the scenario (A) above, a collision could be avoided by using the explicitly agreed  convention (traffic rule) that cars should always swerve to the right, whereas in (B), everyone could go to a famous meeting spot in the city established by a social convention, e.g., the main entrance of the main railway station.
The main entrance of the main railway station is an example of a possible \emph{focal point} (Schelling \cite{Schelling60})
that sticks out and therefore may be likely to emerge as, for example, an obvious meeting point.

In most of this paper we assume that players share \emph{no conventions} at all, and we also assume \emph{no
preplay communication}.
Thus, in our basic setting, the players play completely independently of
each other and could be assumed to come from entirely different kinds of cultures, or even from different galaxies for that matter. 
The principal question in this paper is the following:

\begin{quote}
\emph{What kinds of reasoning can be accepted as purely rational, based on no
communication and no conventions? And which
pure win-lose coordination games can be solved by such reasoning?}
\end{quote}

\noindent
Thus we identify ``\emph{purely rational principles}'' that \emph{every} ideally rational player ought to follow in \emph{every} \WLC game. We also study the hierarchy of games solvable by such principles.
For instance, it is intuitively clear that coordination by pure rationality is not possible in the example situations (A) and (B) above. However, we will see that there are many natural coordination scenarios in which it seems clear that rational players can coordinate successfully by following purely rational reasoning principles, without preplay communication or conventions.

Initially, we assume only \emph{individual rationality}, i.e., that every player acts rationally with the aim to win the game but without assumptions about the other players' rationality.
%
%
Later we additionally assume \emph{common belief in rationality}, i.e., that every player is individually rational and that it is commonly believed amongst the players that every player is rational.

Towards the end of the article, we
move from pure rationality to the setting with conventions. Thus, as 
conventions could indeed arise from preplay communication, we now move to the setting where the
assumption of ``\emph{no preplay communication}'' is essentially banished.\footnote{In particular, we are mainly interested in the
scenario where the players can adopt conventions (for example through negotiations) \emph{only before} being 
presented with the particular game to play.}  
We consider a special class of conventions
that we call \emph{structural conventions}. These are conventions only based on
structural properties of the game
(i.e., essentially properties invariant
under game isomorphisms)
rather than on ad hoc
features, such as names of the choices in the game.


All through the paper we make the following assumptions.
\begin{itemize}
\item We only consider complete information games, i.e., the game structure is common knowledge amongst the players.

\item It is common knowledge amongst the players that
they all have the same goal, which is to win, i.e., to select together a winning choice profile.
\end{itemize}

The main outcomes of our study are as follows:

\begin{enumerate}
\itemsep = 0pt
\item[(i)] Concerning the scenario with no communication and no conventions, 
we identify several different kinds of reasoning principles and provide justifications for them.
These justifications have varying levels of common acceptability. We observe that the boundary of
purely rational principles and other principles is
highly nontrivial to demarcate\footnote{Schelling shares this view on pure coordination games (see \cite{Schelling60}, p. 283, n. 16).}.  Indeed, the question of what constitutes a purely rational principle
seems open-ended and depends on different background assumptions.

\item[(ii)] On the other hand, the class of  \WLC games that are solvable by using structural conventions is precisely characterized in Section \ref{sec:conventions} in terms of structural properties (involving symmetries) of games.
\end{enumerate}

There exist a wide range of works closely related to the current paper. Indeed, coordination and rationality are natural and relevant topics and have thus
been studied in various different contexts, e.g., in 
\cite{Bicchieri95},
\cite{gauthier1975coordination},
\cite{GeneserethGR86},
\cite{gilbert1990rationality}, 
\cite{SugdenRationalChoice91}. 
The theory of focal points has been extensively developed in the context of coordination, e.g., in 
 \cite{Schelling60},  
\cite{gilbert1990rationality}, 
\cite{SugdenRationalChoice91}, 
\cite{mehta1994focal}, 
\cite{sugden1995theory}, 
\cite{Janssen2001},
\cite{BinmoreSamuelson2002}, 
\cite{Bardsley2009}. 
%
%
Conventions have also been studied, e.g., in \cite{Lewis69}, \cite{Sugden89},  \cite{gilbert1990rationality}. 
Furthermore, we note the close conceptual relationship of the present study with the notion of \emph{rationalisability} of strategies \cite{Bernheim84}, \cite{Pearce84}, \cite{FudenbergTirole91},  which is particularly important in epistemic game theory. 
We also mention two recent relevant works---related to logic---to which the observations and results in the present paper could be directly applied: in \cite{hawke2017logic}, two-player coordination games are related to a variant of \emph{Coalition Logic}\footnote{In fact, the initial motivation for the present work came from concerns with the semantics of Alternating time temporal logic ATL, extending Coalition Logic.}, and in \cite{Barbero17}, coordination scenarios are analysed with respect to the game-theoretic semantics of \emph{Independence Friendly Logic}.
Finally, the study on structural conventions with respect to `random game graphs' in \cite{KL2017} adopts the definitions that were originally conceived in the work leading to the current paper.

\lorip{
An extended version of this paper, containing more examples and technical details, is available as a companion technical report \cite{LORI-VI-techrep}.
In addition to the theoretical work presented here, we have also run some empirical experiments on people's behaviour in certain \WLC games. One of our tests can be accessed from the link given in the technical report \cite{LORI-VI-techrep}. 
}

%
%
%
%
This paper is a substantially extended version of the conference paper \cite{lori2017}, while  \cite{sr2017} is an extended abstract of a workshop presentation of a part of this work.

\lorip{\vspace{-1mm}}
\section{Pure win-lose coordination games} 
\label{sec:prelim} 

\lorip{\vspace{-1mm}}
\subsection{The setting} 
\lorip{\vspace{-1mm}}

A  \emph{pure win-lose coordination game} $G$ is a strategic form game with $n$ players ($1,\dots,n$) whose available \emph{choices} (\emph{moves}, \emph{actions}) are given by  sets $\{C_i\}_{i\leq n}$. The set of winning \emph{choice profiles} 
is presented by an $n$-ary \emph{winning relation} $W_G$.
For technical convenience and simplification of some definitions, we present these games as \emph{relational structures} (see, e.g., \cite{ebbinghaus}). A formal definition follows. 

\lorip{\vspace{-1mm}}

\begin{definition}
An $n$-player \defstyle{win-lose coordination game (\WLC game)}
is a relational structure $G=(A,C_1,\dots,C_n,W_G)$ where $A$ is a finite
 domain of \defstyle{choices}, each $C_i\neq\emptyset$ is a unary predicate (i.e., a subset of $A$), representing the choices of player $i$, s.t. $C_1\cup\cdots\cup C_n=A$,  
and $W_G$ is an $n$-ary relation in $A$ such that $W_G\subseteq C_1\times\cdots\times C_n$.
For technical convenience, here we also  assume that the players have pairwise disjoint choice sets, i.e., 
$C_i\cap C_j=\emptyset$ for every  $i,j\leq n$ such that $i\neq j$.
A tuple $\sigma\in C_1\times\cdots\times C_n$ is called a \defstyle{choice profile} for $G$
and the choice profiles in $W_G$ are called \defstyle{winning choice profiles}.
\end{definition}

We use the following terminology for any \WLC game $G=(A,C_1,\dots,C_n,W_G)$.

\lorip{\vspace{-1mm}}

\begin{itemize}[leftmargin=*]
\techrep{
\item The \defstyle{losing relation of $G$} is the relation $L_{G} := C_1\times\cdots\times C_n\setminus W_G$. A choice profile $\sigma\in L_G$ is called \defstyle{losing}.

\smallskip
}

\techrep{
\item The \defstyle{complementary game of $G$} is the game $\overline{G}:=(A,C_1,\dots,C_n,L_G)$.

\smallskip
}

\item Let $A_i\subseteq C_i$ for every $i\leq n$. The
\defstyle{restriction} of $G$ to $(A_1,\dots,A_n)$ is the
game \\ $G\upharpoonright (A_1, \dots, A_n):=(A_1\cup\cdots\cup A_n,\,A_1,\dots,A_n,\,W_G\upharpoonright A_1\times\cdots\times A_n)$. 

\smallskip

\item For every choice $c\in C_{i}$ of a player $i$, the \defstyle{winning extension of $c$ in $G$} is the set $W^i_{G}(c)$ of all tuples $\tau\in C_1\times\cdots\times C_{i-1}\times C_{i+1} \times\cdots\times C_n$ such that the choice profile obtained from $\tau$ by adding $c$ to the $i$-th position is winning.
\techrep{We define the \defstyle{losing extension of $c$ in $G$} analogously.}
 
\smallskip
 
\item A choice $c\in C_{i}$  of a player $i$ is \defstyle{(surely) winning}, respectively  \defstyle{(surely) losing}, if 
it is guaranteed to produce a winning (respectively  losing) choice profile regardless of what choices the other player(s) make. 
Note that $c$ is a winning choice iff $W^i_{G}(c)=C_1\times\cdots\times C_{i-1}\times C_{i+1}\times\cdots\times C_n$. Similarly, $c$ is a losing choice iff $W^i_G(c)=\emptyset$.

\end{itemize}

\lorip{\vspace{-1mm}}

\begin{example}\label{ex: 3-player game}
We present here a 3-player coordination story which will be used as a running example hereafter.
Three robbers, Casper, Jesper and Jonathan\footnote{This example is based on the children's book \emph{When the Robbers Came to Cardamom Town} by Thorbj\o{}rn Egner, featuring the characters  Casper, Jesper and Jonathan.}  
 are planning to quickly steal a cake from the bakery of Cardamom Town while the baker is out. They have two possible plans to enter the bakery: either (a) to break in through the front door or (b) to sneak in through a dark open basement. For (a) they need a \emph{crowbar} and for (b) a \emph{lantern}. The baker keeps the cake on top of a high cupboard, and the robbers can only reach it by using a \emph{ladder}.

When approaching the bakery, Casper is carrying a crowbar, Jesper is carrying a ladder and Jonathan is carrying a lantern. However, the robbers cannot agree whether they should follow plan (a) or plan (b). While the robbers are quarreling, suddenly Constable Bastian appears and the robbers all flee to different directions. After this the robbers have to individually decide whether to go to the front door (by plan (a)) or to the basement entrance (by plan (b)). They must take the right decision fast before the baker returns.

The scenario described here can naturally be modeled as a \WLC game. We relate Casper, Jesper and Jonathan with players 1, 2 and 3, respectively. Each player $i$ has two choices $a_i$ and $b_i$ which correspond to either going to the front door or to the basement entrance, respectively. The robbers succeed in obtaining the cake if both Casper and Jesper go to the front door (whence it does not matter what Jonathan does). Or, alternatively, they succeed if both Jonathan and Jesper go to the basement (whence the choice of Casper is irrelevant). This coordination scenario corresponds to the following \WLC game $G^* = (\{a_1,b_1,a_2,b_2,a_3,b_3\}, C_1, C_2, C_3, W_{G^*})$, where for each player $i$, $C_i=\{a_i,b_i\}$ and $W_{G^*}=\{(a_1,a_2,a_3),(a_1,a_2,b_3),(a_1,b_2,b_3),(b_1,b_2,b_3)\}$. (For a graphical presentation of this game, see Example~\ref{ex: game graph} below.)
\end{example}

\subsection{Presenting \WLC games as hypergraphs} 
 
The $n$-ary winning relation $W_G$ of an $n$-player \WLC game $G$ defines a \emph{hypergraph} on the set of all choices.
We give visual presentations of hypergraphs corresponding to \WLC games as follows. The choices of each player are displayed as columns of nodes starting from the choices of player 1 on the left and
ending with the column with choices of player $n$. The winning relation consists of lines that go through some choice of each player\footnote{In pictures these lines can be drawn in different styles, to set them apart.}.  
This kind of graphical presentation of a \WLC game $G$ will be called a \emph{game graph 
(drawing) 
of $G$}.
(Note that game graphs of 2-player \WLC games are simply bipartite graphs.)

\vspace{5mm}

\begin{example}\label{ex: game graph}
The \WLC game $G^*$ in Example~\ref{ex: 3-player game} has the following game graph:

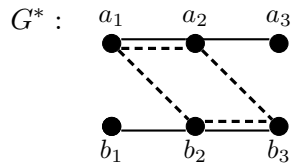
\vspace{-2mm}

\begin{center}
\begin{tikzpicture}[scale=1.3,choice/.style={draw, circle, fill=black!100, inner sep=2.4pt}]
	\node at (-0.9,1.3) {$G^*:$};
	\node at (0,0) [choice] (00) {};
	\node at (0,0) [choice] (00) {};
	\node at (1,0) [choice] (10) {};
	\node at (2,0) [choice] (20) {};
	\node at (0,1) [choice] (01) {};
	\node at (1,1) [choice] (11) {};
	\node at (2,1) [choice] (21) {};
	\draw[thick] ([yshift=-0.5mm]00.east) to ([yshift=-0.5mm]10.west);
	\draw[very thick,densely dashed] ([yshift=-0.6mm]01.east) to ([yshift=-0.6mm]11.west);
	\draw[thick] ([yshift=0.5mm]01.east) to ([yshift=0.5mm]11.west);
	\draw[thick] ([yshift=-0.5mm]10.east) to ([yshift=-0.5mm]20.west);
	\draw[very thick,densely dashed] ([yshift=0.6mm]10.east) to ([yshift=0.6mm]20.west);
	\draw[thick] ([yshift=0.5mm]11.west) to ([yshift=0.5mm]21.west);
	\draw[very thick,densely dashed] (01) to (10);
	\draw[very thick,densely dashed] (11) to (20);
	\node at (0,1.3) {\small $a_1$};
	\node at (1,1.3) {\small $a_2$};
	\node at (2,1.3) {\small $a_3$};
	\node at (0,-0.3) {\small $b_1$};
	\node at (1,-0.3) {\small $b_2$};
	\node at (2,-0.3) {\small $b_3$};
\end{tikzpicture}
\end{center}
\end{example}

\vspace{-2mm}

\lorip{
We now introduce a uniform notation for certain simple classes of \WLC games. 
}
\techrep{
We now define several simple types of \WLC games and introduce a uniform notation for them. Since 
no names of choices are given for these games, each game given here actually corresponds to a class of games with the same structure. However, in this paper we usually consider such games to be equivalent\footnote{If a player reasons by pure rationality, the names of the choices should not have an effect on that player's reasoning. We will discuss further this issue later on.}.
}
Let $k_1,\dots, k_n\in\mathbb{N}$. 
\begin{itemize}[leftmargin=*]
\item $G(k_1\times \cdots \times k_n)$ is the $n$-player \WLC game where the player $i$ has $k_i$ choices and the winning relation is the \emph{universal relation} $C_1\times \cdots \times C_n$. 

\smallskip

\item $G(\overline{k_1\times \cdots \times k_n})$ 
is the $n$-player \WLC game where the player $i$ has $k_i$ choices and the winning relation is the \emph{empty relation}. 
%
Note that with this notation $G(\overline{k_1\times \cdots \times k_n})=\overline{G(k_1\times \cdots \times k_n)}$.
%
Some examples: 

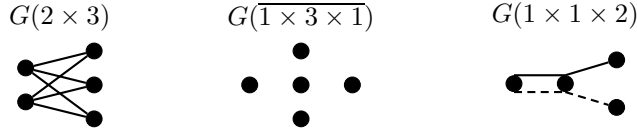
\vspace{-2mm}

\begin{center}
\begin{tikzpicture}[scale=0.5,choice/.style={draw, circle, fill=black!100, inner sep=2.1pt}]
	\node at (1,3) {\small\bf $G(2\times 3)$};	
	\node at (0,0.5) [choice] (00) {};
	\node at (0,1.5) [choice] (01) {};
	\node at (2,0) [choice] (20) {};
	\node at (2,1) [choice] (21) {};
	\node at (2,2) [choice] (22) {};
	\draw[thick] (00) to (20);	
	\draw[thick] (00) to (21);
	\draw[thick] (00) to (22);
	\draw[thick] (01) to (20);	
	\draw[thick] (01) to (21);
	\draw[thick] (01) to (22);
\end{tikzpicture}
\qquad\quad
\begin{tikzpicture}[scale=0.5,choice/.style={draw, circle, fill=black!100, inner sep=2.1pt}]
	\node at (1.5,3) {\small\bf $G(\overline{1\times 3\times 1})$};	
	\node at (0,1) [choice] {};	
	\node at (1.5,0) [choice] {};
	\node at (1.5,1) [choice] {};
	\node at (1.5,2) [choice] {};
	\node at (3,1) [choice] {};
\end{tikzpicture}
\qquad\quad
\begin{tikzpicture}[scale=0.5,choice/.style={draw, circle, fill=black!100, inner sep=2.1pt}]
	\node at (1.5,3) {\small\bf $G(1\times 1\times 2)$};	
	\node at (0,1) [choice] (1) {};	
	\node at (1.5,1) [choice] (2) {};
	\node at (3,1.7) [choice] (3a) {};
	\node at (3,0.3) [choice] (3b) {};
	\node at (0,0) (dummy) {};
	\draw[thick] (1.north) to (2.north);
	\draw[thick, densely dashed] (1.south) to (2.south);
	\draw[thick] (2.north) to (3a);
	\draw[thick, densely dashed] (2.south) to (3b);
\end{tikzpicture}
\end{center}

\vspace{-1mm}

\item Let $k\in\mathbb{N}$. 
We write $G(Z_k)$ for the \emph{2-player} \WLC game in which both players have $k$ choices and the winning relation forms a single path that goes through all the choices (see below for an example). Similarly, $G(O_k)$, where $k\geq 2$, denotes the $2$-player \WLC game where the winning relation forms a $2k$-cycle that goes through all the choices. These are exemplified by the following:

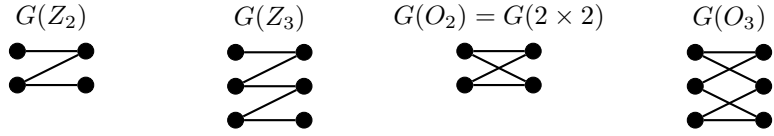
\vspace{-2mm}

\begin{center}
\begin{tikzpicture}[scale=0.5,choice/.style={draw, circle, fill=black!100, inner sep=2.1pt}]
	\node at (1,2) {\small\bf $G(Z_2)$};	
	\node at (0,0) [choice] (00) {};
	\node at (2,0) [choice] (20) {};
	\node at (0,1) [choice] (01) {};	
	\node at (2,1) [choice] (21) {};
	\node at (0,-1) {};
	\draw[thick] (01) to (21);
	\draw[thick] (00) to (20);	
	\draw[thick] (00) to (21);
\end{tikzpicture}
\qquad\qquad
\begin{tikzpicture}[scale=0.5,choice/.style={draw, circle, fill=black!100, inner sep=2.1pt}]
	\node at (1,3) {\small\bf $G(Z_3)$};	
	\node at (0,0) [choice] (00) {};
	\node at (2,0) [choice] (20) {};
	\node at (0,1) [choice] (01) {};	
	\node at (2,1) [choice] (21) {};
	\node at (0,2) [choice] (02) {};	
	\node at (2,2) [choice] (22) {};
	\draw[thick] (01) to (21);
	\draw[thick] (00) to (20);	
	\draw[thick] (00) to (21);
	\draw[thick] (02) to (22);
	\draw[thick] (01) to (22);
\end{tikzpicture}
\qquad
\begin{tikzpicture}[scale=0.5,choice/.style={draw, circle, fill=black!100, inner sep=2.1pt}]
	\node at (1,2) {\small\bf $G(O_2)=G(2\times 2)$};	
	\node at (0,0) [choice] (00) {};
	\node at (2,0) [choice] (20) {};
	\node at (0,1) [choice] (01) {};	
	\node at (2,1) [choice] (21) {};
	\node at (0,-1) {};
	\draw[thick] (01) to (21);
	\draw[thick] (01) to (20);
	\draw[thick] (00) to (20);	
	\draw[thick] (00) to (21);
\end{tikzpicture}
\qquad
\begin{tikzpicture}[scale=0.5,choice/.style={draw, circle, fill=black!100, inner sep=2.1pt}]
	\node at (1,3) {\small\bf $G(O_3)$};	
	\node at (0,0) [choice] (00) {};
	\node at (2,0) [choice] (20) {};
	\node at (0,1) [choice] (01) {};	
	\node at (2,1) [choice] (21) {};
	\node at (0,2) [choice] (02) {};	
	\node at (2,2) [choice] (22) {};
	\draw[thick] (00) to (20);	
	\draw[thick] (00) to (21);
	\draw[thick] (02) to (22);
	\draw[thick] (01) to (22);
	\draw[thick] (02) to (21);
	\draw[thick] (01) to (20);
\end{tikzpicture}
\end{center}

\vspace{-1mm}

\item Suppose that $G(A)$ and $G(B)$ have been defined, both having the same number of players.
Then $G(A + B)$ is the \emph{disjoint union} of $G(A)$ and $G(B)$, i.e.,
the game obtained by assigning to each player a disjoint union of her choices in 
 $G(A)$ and $G(B)$, and where the winning relation for $G(A + B)$ is the union of the winning relations in $G(A)$ and $G(B)$. Some examples:


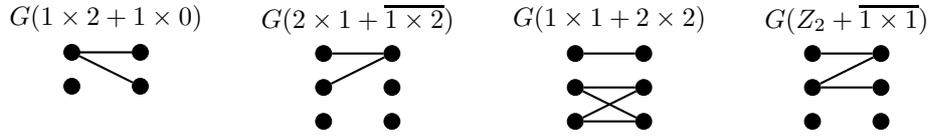
\begin{center}
\begin{tikzpicture}[scale=0.5,choice/.style={draw, circle, fill=black!100, inner sep=2.1pt}]
	\node at (1,2) {\small\bf $G(1\times 2 + 1\times 0)$};	
	\node at (0,0) [choice] (00) {};
	\node at (2,0) [choice] (20) {};
	\node at (0,1) [choice] (01) {};	
	\node at (2,1) [choice] (21) {};
	\node at (2,-1) {};
	\draw[thick] (01) to (21);
	\draw[thick] (01) to (20);	
\end{tikzpicture}
\quad
\begin{tikzpicture}[scale=0.5,choice/.style={draw, circle, fill=black!100, inner sep=2.1pt}]
	\node at (1,3) {\small\bf $G(2\times 1 + \overline{1\times 2})$};	
	\node at (0,0) [choice] (00) {};
	\node at (2,0) [choice] (20) {};
	\node at (0,1) [choice] (01) {};	
	\node at (2,1) [choice] (21) {};
	\node at (0,2) [choice] (02) {};	
	\node at (2,2) [choice] (22) {};
	\draw[thick] (02) to (22);	
	\draw[thick] (01) to (22);
\end{tikzpicture}
\quad
\begin{tikzpicture}[scale=0.5,choice/.style={draw, circle, fill=black!100, inner sep=2.1pt}]
	\node at (1,3) {\small\bf $G(1\times 1 + 2\times 2)$};	
	\node at (0,0) [choice] (00) {};
	\node at (2,0) [choice] (20) {};
	\node at (0,1) [choice] (01) {};	
	\node at (2,1) [choice] (21) {};
	\node at (0,2) [choice] (02) {};	
	\node at (2,2) [choice] (22) {};
	\draw[thick] (01) to (21);
	\draw[thick] (01) to (20);
	\draw[thick] (00) to (20);	
	\draw[thick] (00) to (21);
	\draw[thick] (02) to (22);	
\end{tikzpicture}
\quad
\begin{tikzpicture}[scale=0.5,choice/.style={draw, circle, fill=black!100, inner sep=2.1pt}]
	\node at (1,3) {\small\bf $G(Z_2 + \overline{1\times 1})$};	
	\node at (0,0) [choice] (00) {};
	\node at (2,0) [choice] (20) {};
	\node at (0,1) [choice] (01) {};	
	\node at (2,1) [choice] (21) {};
	\node at (0,2) [choice] (02) {};	
	\node at (2,2) [choice] (22) {};
	\draw[thick] (01) to (21);
	\draw[thick] (01) to (22);
	\draw[thick] (02) to (22);	
\end{tikzpicture}
\end{center}

\vspace{-1mm}

\item Let $m\in\mathbb{N}$. Then  $G(mA):=G(A + \cdots + A)$ ($m$ times). Examples: 

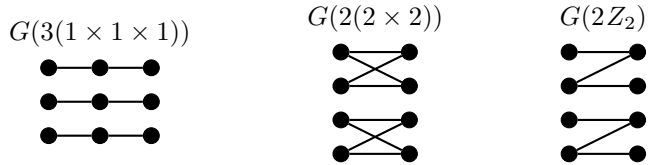
\vspace{-2mm}

\begin{center}
\begin{tikzpicture}[scale=0.5,choice/.style={draw, circle, fill=black!100, inner sep=2.1pt}]
	\node at (1.5,3) {\small\bf $G(3(1\times 1\times 1))$};	
	\node at (0,0) [choice] (1c) {};
	\node at (1.5,0) [choice] (2c) {};
	\node at (3,0) [choice] (3c) {};
	\node at (0,1) [choice] (1b) {};	
	\node at (1.5,1) [choice] (2b) {};
	\node at (3,1) [choice] (3b) {};
	\node at (0,2) [choice] (1a) {};	
	\node at (1.5,2) [choice] (2a) {};
	\node at (3,2) [choice] (3a) {};
	\node at (0,-0.5) (dummy) {};
	\draw[thick] (1a) to (2a);	
	\draw[thick] (2a) to (3a);
	\draw[thick] (1b) to (2b);	
	\draw[thick] (2b) to (3b);
	\draw[thick] (1c) to (2c);	
	\draw[thick] (2c) to (3c);
\end{tikzpicture}
\qquad\quad
\begin{tikzpicture}[scale=0.5,choice/.style={draw, circle, fill=black!100, inner sep=2.1pt}]
	\node at (1,4) {\small\bf $G(2(2\times 2))$};	
	\node at (0,0) [choice] (00) {};
	\node at (2,0) [choice] (20) {};
	\node at (0,1) [choice] (01) {};	
	\node at (2,1) [choice] (21) {};
	\node at (0,2) [choice] (02) {};	
	\node at (2,2) [choice] (22) {};
	\node at (0,3) [choice] (03) {};	
	\node at (2,3) [choice] (23) {};
	\draw[thick] (01) to (21);
	\draw[thick] (00) to (20);	
	\draw[thick] (00) to (21);
	\draw[thick] (01) to (20);	
	\draw[thick] (02) to (22);	
	\draw[thick] (03) to (23);
	\draw[thick] (03) to (22);
	\draw[thick] (02) to (23);
\end{tikzpicture}
\qquad\quad
\begin{tikzpicture}[scale=0.5,choice/.style={draw, circle, fill=black!100, inner sep=2.1pt}]
	\node at (1,4) {\small\bf $G(2\hspace{0,2mm}Z_2)$};	
	\node at (0,0) [choice] (00) {};
	\node at (2,0) [choice] (20) {};
	\node at (0,1) [choice] (01) {};	
	\node at (2,1) [choice] (21) {};
	\node at (0,2) [choice] (02) {};	
	\node at (2,2) [choice] (22) {};
	\node at (0,3) [choice] (03) {};	
	\node at (2,3) [choice] (23) {};
	\draw[thick] (01) to (21);
	\draw[thick] (00) to (20);	
	\draw[thick] (00) to (21);
	\draw[thick] (02) to (22);	
	\draw[thick] (03) to (23);
	\draw[thick] (02) to (23);
\end{tikzpicture}
\end{center}

\vspace{-1mm}

\item Recall our ``regrouping scenario'' (B) from the introduction. If there are $n$ people in the group and there are $m$ possible meeting spots in the city, then the game is of the form $G(m(1^n))$, where $1^n:=1\times \cdots \times 1$ ($n$ times).
\end{itemize}

\subsection{On choice domination and Nash equilibria in \WLC games} 
 
Given a \WLC game $G=(A,C_1,\dots,C_n,W_G)$, we say that a choice $c\in C_{i}$ is \defstyle{at least as good as} (respectively, \defstyle{better than}) a choice $c'\in C_{i}$ if $W^i_{G}(c') \subseteq W^i_{G}(c)$ (respectively, $W^i_{G}(c') \subsetneq W^i_{G}(c)$).  
A choice $c\in C_{i}$ is \defstyle{optimal} for a player $i$ if it is at least as good as any other choice of $i$.   

Note that a choice $c\in C_{i}$ is better than a choice $c'\in C_{i}$ precisely when $c$ \emph{weakly dominates} $c'$ in the usual game-theoretic sense (see e.g. \cite{Leyton2008}, 
\cite{Peters}).
Respectively, a choice $c\in C_i$ is an optimal choice of player $i$ when it is a \emph{weakly dominant} choice (i.e., a choice that weakly dominates all other choices).

Note that  $c$ \emph{strictly dominates} $c'$ (\textit{ibid.}) if and only if $c$ is a surely winning choice and $c'$ is a surely losing choice. 
%
%
Thus, strict domination between choices is a too strong concept in \WLC games. 
Also, the concept of \emph{Nash equilibrium} (NE) for choice profiles is not very useful here, because not only every winning choice profile is a NE, but so is also every losing choice profile  which no player can unilaterally convert to a winning one by changing their choice\footnote{In the special case of two-player games, this latter case amounts to choice profiles consisting only of (surely) losing choices.}. 
For instance, in the game displayed below, the Nash equilibria are not only all 4 winning profiles, $(a_{1},a_{2},a_{3})$, $(b_{1},b_{2},c_{3})$, 
$(c_{1},b_{2},b_{3})$, and $(c_{1},c_{2},c_{3})$, but also the losing profiles 
$(a_{1},c_{2},b_{3})$ and  $(b_{1},c_{2},a_{3})$. 

\vspace{-3mm}

\begin{center}
\begin{tikzpicture}[scale=1,choice/.style={draw, circle, fill=black!100, inner sep=2.1pt}]
	\node at (0,0) [choice] (00) {};
	\node at (1,0) [choice] (10) {};
	\node at (2,0) [choice] (20) {};
	\node at (0,1) [choice] (01) {};
	\node at (1,1) [choice] (11) {};
	\node at (2,1) [choice] (21) {};	
	\node at (0,2) [choice] (02) {};
	\node at (1,2) [choice] (12) {};
	\node at (2,2) [choice] (22) {};
	\draw[thick] ([yshift=-0.5mm]00.east) to ([yshift=-0.5mm]10.west);
%
%
	\draw[thick] ([yshift=0.5mm]01.east) to ([yshift=0.5mm]11.west);
	\draw[thick] ([yshift=-0.5mm]10.east) to ([yshift=-0.5mm]20.west);
	\draw[very thick,densely dashed] ([yshift=0.5mm]11.west) to ([yshift=0.5mm]21.west);
%
%
	\draw[thick] (02.east) to (12.west);
	\draw[thick] (12.west) to (22.west);	
	\draw[very thick,densely dashed] (00) to (11);
%
%
	\draw[thick] (11) to (20);
%
	\node at (0,2.3) {\small $a_1$};
	\node at (1,2.3) {\small $a_2$};
	\node at (2,2.3) {\small $a_3$};
	\node at (0,1.3) {\small $b_1$};
	\node at (1,1.3) {\small $b_2$};
	\node at (2,1.3) {\small $b_3$};
	\node at (0,-0.3) {\small $c_1$};
	\node at (1,-0.3) {\small $c_2$};
	\node at (2,-0.3) {\small $c_3$};
\end{tikzpicture}
\end{center}
We will not make use of Nash equilibria further in this work.

\subsection{Symmetries of \WLC games and structural protocols}

A \defstyle{protocol} is a mapping $\Sigma$ that
assigns to every pair $(G,i)$, where $G$ is a \WLC game and $i$ a player in $G$, a nonempty set $\Sigma(G,i)\subseteq C_i$ of choices.
%
Thus, a protocol describes a global nondeterministic strategy for playing any \WLC game in the role of any player. Intuitively, such protocol
represents a global mode of acting in
any situation that involves playing \WLC games. Hence, protocols can be
informally regarded as global ``reasoning styles" or ``behaviour modes". Thus, a protocol can also be identified with an agent who acts according to that protocol in all situations that involve playing different \WLC games in different player roles.
Note, however, that this correspondence is not bijective, as several agents may behave according to the same protocol and some protocols might not correspond to the behaviour of any (actual) agent.

Assuming a setting  with no special conventions or preplay communication, a protocol will only take into account the \emph{structural properties} of the game and its winning relation. Thus, it is reasonable to assume that the names of the choices and the names (or ordering) of the players should be of no relevance. In this section we make this issue precise.
\lorip{
(For more details and examples, see~\cite{LORI-VI-techrep}.)
}

\begin{definition}\label{def: choice renaming}
An isomorphism\footnote{Isomorphism is defined as usual for relational structures (see, e.g., \cite{ebbinghaus}).} 
between games $G$ and $G'$ is called a \defstyle{choice-renaming}. An automorphism of $G$ is called a \defstyle{choice-renaming of $G$}. 

Let $G=(A,C_1,\dots,C_n,W_G)$ be a \WLC game. For a player $i$, we say that the choices $c,c'\in C_i$ are \defstyle{$i$-equivalent}, denoted by $c\simeq_i c'$, if there is a choice-renaming of $G$ that maps $c$ to $c'$. 
For each $i\leq n$, the relation $\simeq_i$ is an equivalence relation on the set $C_i$. We denote the equivalence class of $c\in C_i$ by $\llbracket c\rrbracket_i$.
\end{definition}

\techrep{
Supposing that a player $i$ does not use names or labels of her choices (or she has no preferences over them), then she should be \emph{indifferent} about the choices that are in the same equivalence class.
}

\techrep{
\begin{example}
Let $A=\{a_1,b_1,c_1,a_2,b_2,c_2\}$ and $A=\{a_1',b_1',c_1',a_2',b_2',c_2'\}$. Consider the \WLC games $G$ and $G'$ whose game graphs are given below.
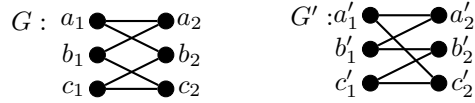
\begin{center}
\begin{tikzpicture}[scale=0.5,choice/.style={draw, circle, fill=black!100, inner sep=2.1pt}]
	\node at (-2,2) {\small $G:$};	
	\node at (0,0) [choice] (00) {};
	\node at (2,0) [choice] (20) {};
	\node at (0,1) [choice] (01) {};	
	\node at (2,1) [choice] (21) {};
	\node at (0,2) [choice] (02) {};	
	\node at (2,2) [choice] (22) {};
	\draw[thick] (00) to (20);	
	\draw[thick] (00) to (21);
	\draw[thick] (02) to (22);
	\draw[thick] (01) to (22);
	\draw[thick] (02) to (21);
	\draw[thick] (01) to (20);
	\node at (-0.7,2) {\small $a_1$};
	\node at (-0.7,1) {\small $b_1$};
	\node at (-0.7,0) {\small $c_1$};
	\node at (2.7,2) {\small $a_2$};
	\node at (2.7,1) {\small $b_2$};
	\node at (2.7,0) {\small $c_2$};
\end{tikzpicture}
\quad\quad
\begin{tikzpicture}[scale=0.5,choice/.style={draw, circle, fill=black!100, inner sep=2.1pt}]
	\node at (-2,2) {\small $G':$};	
	\node at (0,0) [choice] (00) {};
	\node at (2,0) [choice] (20) {};
	\node at (0,1) [choice] (01) {};	
	\node at (2,1) [choice] (21) {};
	\node at (0,2) [choice] (02) {};	
	\node at (2,2) [choice] (22) {};
	\draw[thick] (01) to (21);
	\draw[thick] (00) to (20);	
	\draw[thick] (00) to (21);
	\draw[thick] (02) to (22);
	\draw[thick] (01) to (22);
	\draw[thick] (02) to (20);
	\node at (-0.7,2) {\small $a_1'$};
	\node at (-0.7,1) {\small $b_1'$};
	\node at (-0.7,0) {\small $c_1'$};
	\node at (2.7,2) {\small $a_2'$};
	\node at (2.7,1) {\small $b_2'$};
	\node at (2.7,0) {\small $c_2'$};
\end{tikzpicture}
\end{center}
A function $\pi:A\rightarrow A'$, which maps $b_2$ to $c_2'$, $c_2$ to $b_2'$, and $x$ to $x'$ for all the other choices $x\in A$, is choice-renaming from $G$ to $G'$.
Note that actually both $G$ and $G'$ are of the form $G(O_3)$.
A function that maps $a_i$ to $b_i$, $b_i$ to $c_i$, and $c_i$ to $a_i$ (for $i\in\{1,2\}$) is a choice-renaming of $G$. Therefore $a_1\simeq_1 b_1\simeq_1 c_1$ and $a_2\simeq_2 b_2\simeq_2 c_2$.
\end{example}
}

\begin{definition}\label{def: full renaming}
Consider $n$-player \WLC games $$G=(A,C_1,\dots,C_n,W_G)
\text{ and }G'=(A,C_1',\dots,C_n',W_G').$$ A permutation $\beta:\{1,...,n\}\rightarrow\{1,...,n\}$ is
called a \defstyle{player-renaming}
between $G$ and $G'$ if the following conditions hold:
\begin{enumerate}
\item[(1)]
$C_{\beta(i)} = C_{i}'$ for each $i\leq n$.\vspace{0.7mm}
\item[(2)]
$W_G' = \{\, (c_{\beta(1)},\dots,c_{\beta(n)})\, |\, (c_{1},\dots,c_{n})\in W_G\, \}$.
\end{enumerate}
\end{definition}

If there is a player-renaming between two \WLC games, the games are essentially the same, the only difference being the ordering of the players.
\techrep{Furthermore, the game graph of $G'$ is simply obtained by permuting the columns of the game graph of $G$.
}

\techrep{
\begin{example}
Consider the following \WLC games:
\begin{center}
\begin{tikzpicture}[scale=0.5,choice/.style={draw, circle, fill=black!100, inner sep=2.1pt}]
	\node at (-2,2) {\small\bf $G$:};	
	\node at (0,0) [choice] (00) {};
	\node at (0,1.5) [choice] (01) {};
	\node at (2,0) [choice] (20) {};
	\node at (2,1) [choice] (21) {};
	\node at (2,2) [choice] (22) {};
	\draw[thick] (00) to (20);	
	\draw[thick] (01) to (20);	
	\draw[thick] (01) to (21);
	\draw[thick] (01) to (22);
	\node at (-0.7,1.5) {\small $a_1$};
	\node at (-0.7,0) {\small $b_1$};
	\node at (2.7,2) {\small $a_2$};
	\node at (2.7,1) {\small $b_2$};
	\node at (2.7,0) {\small $c_2$};
\end{tikzpicture}
\qquad
\begin{tikzpicture}[scale=0.5,choice/.style={draw, circle, fill=black!100, inner sep=2.1pt}]
	\node at (-2,2) {\small\bf $G'$:};	
	\node at (2,0) [choice] (00) {};
	\node at (2,1.5) [choice] (01) {};
	\node at (0,0) [choice] (20) {};
	\node at (0,1) [choice] (21) {};
	\node at (0,2) [choice] (22) {};
	\draw[thick] (00) to (20);	
	\draw[thick] (01) to (20);	
	\draw[thick] (01) to (21);
	\draw[thick] (01) to (22);
	\node at (2.7,1.5) {\small $a_1$};
	\node at (2.7,0) {\small $b_1$};
	\node at (-0.7,2) {\small $a_2$};
	\node at (-0.7,1) {\small $b_2$};
	\node at (-0.7,0) {\small $c_2$};
\end{tikzpicture}
\end{center}
A permutation $\beta$, which swaps 1 and 2, is a player-renaming between $G$ and $G'$.
\end{example}
}

\begin{definition}\label{deftv}
Consider \WLC games $G$ and $G'$. A pair $(\beta,\pi)$ is a \defstyle{full renaming} between $G$ and $G'$ if there is a \WLC game $G''$ such that $\beta$ is a player-renaming between $G$ and $G''$ and $\pi$ is a choice-renaming between $G''$ and $G'$. 
If $G$ and $G'$ have the same domain $A$, we say that $(\beta,\pi)$ is a \defstyle{full renaming of $G$}.
We say that choices $c\in C_i$ and $c'\in C_j$ in the same game are \defstyle{structurally equivalent}, denoted by $c\sim c'$, if there is a full renaming $(\beta,\pi)$ of $G$ such that $\beta(i)=j$ and $\pi(c)=c'$.
It is quite easy to see that $\sim$ is an equivalence relation on the set $A$ of all choices. We denote the equivalence class of a choice $c$ by $[c]$.
\end{definition}

\techrep{
We also make the following observations: 
\begin{itemize}[leftmargin=*]
\item If $c\simeq_i c'$ for some $i$, then also $c\sim c'$.
\item Suppose that there is a sequence $G_1,\dots,G_n$ of \WLC games such that for every $i$ there is either a choice-renaming or a player-renaming between $G_i$ and $G_{i+1}$. Then it is easy to show that there is a full renaming from $G_1$ to $G_n$.
\end{itemize}
}

\begin{example}
Consider a \WLC game of the form $G(1\times 2 + 2\times 1)$:


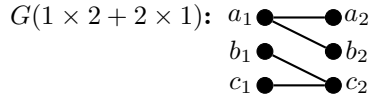
\begin{center}
\begin{tikzpicture}[scale=0.5,choice/.style={draw, circle, fill=black!100, inner sep=2.1pt}]
	\node at (-4.5,2) {};	
	\node at (0,0) [choice] (00) {};
	\node at (2,0) [choice] (20) {};
	\node at (0,1) [choice] (01) {};	
	\node at (2,1) [choice] (21) {};
	\node at (0,2) [choice] (02) {};	
	\node at (2,2) [choice] (22) {};
	\draw[thick] (00) to (20);	
	\draw[thick] (02) to (22);
	\draw[thick] (02) to (21);
	\draw[thick] (01) to (20);
	\node at (-0.7,2) {\small $a_1$};
	\node at (-0.7,1) {\small $b_1$};
	\node at (-0.7,0) {\small $c_1$};
	\node at (2.7,2) {\small $a_2$};
	\node at (2.7,1) {\small $b_2$};
	\node at (2.7,0) {\small $c_2$};
\end{tikzpicture}
\end{center}
\techrep{
Let $\beta$ be the permutation which swaps (players) 1 and 2, and let $\pi$ be the bijection
$$\{(a_1,c_2),(b_1,b_2),(c_1,a_2),(a_2,c_1),(b_2,b_1),(c_2,a_1)\}.$$ Now the pair $(\beta,\pi)$ is a full renaming of $G(1\times 2 + 2\times 1).$
}
It is easy to see that $\simeq_1$ has the equivalence classes $\{a_1\}$ and $\{b_1,c_1\}$, and similarly, $\simeq_2$ has equivalence classes $\{c_2\}$ and $\{a_2,b_2\}$.
Furthermore, $\sim$ has the equivalence classes $\{a_1,c_2\}$ and $\{b_1,c_1,a_2,b_2\}$. Likewise, in the game $G^*$ from Example~\ref{ex: 3-player game} the relation $\sim$ has equivalence classes $\{a_1,b_3\}$, $\{b_1,a_3\}$, $\{a_2,b_2\}$.

\end{example}

We say that a protocol $\Sigma$ is \defstyle{structural} if it is ``indifferent'' with respect to full renamings, which means that, given any \WLC games $G$, $G'$ for which there exists a full renaming $(\beta,\pi)$ between $G$ and $G'$, for any $i$ and any choice $c\in C_i$, it must hold that
$$c\in\Sigma(G,i) \, \text{ iff } \, \pi(c)\in\Sigma(G',\beta(i)).$$
Intuitively, this reflects the idea that, when following a structural protocol one
acts independently of the names of choices and names (or ordering) of player roles.\footnote{In the definition of \WLC games, player roles appear as
(naturally ordered) indices $i$. However, this presentation is only for technical convenience, and player roles could instead be called, e.g., ``white'', ``black'' etc.}
Thus, following a structural protocol, one cannot tell the difference between choices that are structurally equivalent.
Hereafter, unless otherwise specified, we only consider structural protocols.

It is worth noting that if we considered a framework where \WLC games were presented so that the names of the choices and players could always be used to
uniquely define an ordering (of the players and their choices),
solving games could be
trivialised by using the pre-negotiated agreement to always choose the lexicographically least tuple from the winning relation. For more on solving coordination games with names or ordering of choices, see~\cite{DBLP:conf/eumas/GorankoKR17}.

\section{Purely rational principles in \WLC games}\label{sectionttt}

In this section we will analyse various principles which players can follow in \WLC games. We will provide justifications for these principles and study which games can be surely won when such principles are followed. Our aim is to characterize which principles are ``purely rational'' in the sense that all ideally rational agents ought to follow them in all \WLC games.

\subsection{Purely rational principles}

By a \defstyle{principle} we mean here \emph{any nonempty class of protocols}. 
Intuitively, these are the protocols ``complying" with that principle.
Hence principles are simply \emph{properties of protocols} and if protocols are regarded as ``reasoning styles", then principles are properties of reasoning. 
Likewise, if protocols are identified with agents who behave according to them, then principles can be seen as norms which agents follow.

Principles that contain only structural protocols are called \defstyle{structural principles}. 
Such principles are properties of structural protocols and therefore they describe behavior which is independent of names of the choices and roles of the players.

We are mainly interested in principles which describe ``purely rational and perfect reasoning'' that involves neither preplay communication nor conventions. The principles corresponding to such reasoning are defined as follows:
\begin{quote}
A principle P is called \emph{purely rational} if P is followed by \emph{all ideally rational agents} in \emph{every} \WLC game.
\end{quote}
An ideally rational agent always reasons in an optimal, faultless way. Note that we \emph{do not give a formal definition} of an ideally rational rational agent. It is taken to be a philosophical primitive that is central to our study, but yet a notion that we \emph{cannot} define formally in this paper\footnote{The reader should note that, because the notion of an ideally rational agent is taken as a conceptual primitive, it is problematic to formally prove that any of the principles defined in the paper are actually purely rational. Still, we try to give strong rational justifications for these principles and argue that many of them should be followed by all ideally rational agents. However, when we proceed to stronger principles, it becomes harder to give the principles a solid rational justification and to argue that they would be purely rational. One of our main philosophical aims in the paper is to demonstrate that it is very difficult to identify the boundary of purely rational principles and other principles.}. 
Note also that, since purely rational principles are indeed followed by all ideally rational agents, such principles are not based on any particular conventions.

We say that a player $i$ \defstyle{follows a principle P} in a \WLC game $G$ if she plays according to \emph{some}  protocol in P. 
Consider for example the following principles:
\vspace{-0,1cm}
\begin{align*}
	&\mathrm{P}_1:=\{\Sigma\mid \Sigma(G,i) \text{ does not contain any surely losing choices when } W_G\neq\emptyset\},\\[-0,05cm]
	&\mathrm{P}_2:=\{\Sigma\mid \Sigma(G,i)\text{ contains all choices } c\in C_i
              \text{ such that }          |W^i_G(c)| \\[-0,1cm]
         &\hspace{1.9cm} \text{ is a prime number; if there are no such choices, }
              \Sigma(G,i) = C_i.\}. \\[-0,7cm]
\end{align*}
If player $i$ follows $\mathrm{P}_1$, then she uses some protocol which never selects surely losing choices, if possible. This seems a principle that any rational agent would follow, so it can be regarded as a purely rational principle.  
Likewise, if player $i$ follows $\mathrm{P}_2$, then she always plays choices whose degree (in the game graph) is a prime number, if possible. This principle seems arbitrary; it could only be some artificial convention, for example.

We say that a \defstyle{principle P solves} a \WLC game $G$ (or \defstyle{$G$ is P-solvable}), if $G$ is won whenever every player follows some protocol that belongs to P. Formally, this means that $\Sigma_1(G,1)\times \cdots\times\Sigma_n(G,n)\subseteq W_G$ for all protocols 
$\Sigma_1,\dots,\Sigma_n\in\mathrm{P}$. 
The class of all P-solvable games is denoted by $s(\mathrm{P})$. 

Hereafter, we will identify (a hierarchy of) principles that can be considered to be purely rational and will analyse the classes of games that they solve.
Since we have argued that purely rational principles should be structural, for every principle P which we define we assume that P only consists of \emph{structural} protocols which satisfy the description of $\mathrm{P}$.

\begin{remark}\label{rem: Intensions}
When defining principles, we give their definition informally by describing some rational mode of behaviour. For example, the non-losing principle $\mathrm{NL}$ in the following section is defined by the sentence ``do not play a losing choice, if possible''. We call this description the \defstyle{intension} of $\mathrm{NL}$, while the \defstyle{extension} of $\mathrm{NL}$ is the actual set of structural protocols that satisfy this intension.
In this paper we describe the intensions in natural language but they could be formalized in e.g. some formal logic.

Note that there might be several different intensions that correspond to the same extension. Also, there may be a principle $\mathrm{P}$ (consisting of arbitrary protocols) for which there is no natural intension. A similar distinction can also be made for protocols: an agent behaving according to his protocol could have some intension that justifies his behaviour, but we may not be able to identify this intension by only looking at the (extension of the) protocol. 
This also bears a link to the concept of \emph{rationalizable strategies} 
\cite{Bernheim84}, \cite{Pearce84}. 
\end{remark}

\begin{remark}
We have defined principles as sets of protocols for mainly conceptual reasons to make a distinction between protocols and their properties. However, for the results of the current paper, we could have defined principles as protocols simply by forming a corresponding ``union protocol''. That is, given a principle $\mathrm{P}$, we can form the protocol $\Sigma_\mathrm{P}$ by defining $\Sigma_\mathrm{P}(G,i)=\bigcup_{\Sigma\in\mathrm{P}}\Sigma(G,i)$, whence we have $s(\mathrm{P})=s(\{\Sigma_\mathrm{P}\})$. 
For the principles defined in this paper, we also happen to have $\Sigma_\mathrm{P}\in\mathrm{P}$, but generally this does not need to be the case. Moreover, this correspondence is not bijective since we may have $\Sigma_\mathrm{P_1}=\Sigma_\mathrm{P_2}$ for some protocols $\mathrm{P_1}\neq\mathrm{P_2}$.
With the current definitions, we can also combine principles in an easy way by simply taking their intersections---that is, a player follows both $\mathrm{P_1}$ and $\mathrm{P_2}$ if and only if (s)he follows the principle $\mathrm{P_1}\cap\mathrm{P_2}$.
\end{remark}

\subsection{Basic individual rationality}\label{ssec: BIR}

Hereafter we describe principles by the properties of protocols that they determine.
We begin by considering the case 
where players are individually rational, but there is no common knowledge
about this being the case. 
It is safe to assume that any individually rational player would follow at least the following principle.

\medskip
\textbf{Fundamental individual rationality (FIR):}

\textsf{Never play a strictly dominated choice\footnote{Recall, that a choice $a$ is strictly dominated by a choice $b$ if the choice $b$ guarantees a strictly higher payoff than the choice $a$ in every play of the game (see e.g. 
\cite{Leyton2008}, \cite{Peters}).}, if possible.} 

\medskip

As noted earlier, strict domination is a very strong concept for \WLC games. Following FIR simply means that a player should never prefer a surely losing choice to a  surely winning one. 
Therefore FIR is a very weak principle that can solve only some quite trivial types of games such as $G(1\times 2 + 1\times 0)$\techrep{\,(See Figure 1)}.
In general, FIR-solvable games have a simple description: at least one of the players has (at least one) surely winning choice, and all non-winning choices of that player are surely losing. \techrep{Thus, for example all games of the form $G(k\times l + m\times 0)$ are FIR-solvable.}

FIR has two natural strengthenings that can still be considered purely rational: 

\begin{enumerate}
\item \defstyle{Non-losing principle (NL):} 
\textsf{Never play a losing choice, if possible.}

\smallskip

\item \defstyle{Sure winning principle (SW):}
\textsf{Always play a winning choice, if possible.} 
\end{enumerate}
\vspace{-1mm}

Since losing choices cannot be winning choices, these principles can naturally be put together by taking their \emph{intersection} (recall here that principles are simply sets of protocols that satisfy the given property).

\medskip
\textbf{Basic individual rationality (BIR):} \; NL $\cap$ SW.
\medskip
 
\noindent
Thus, when following BIR, a player plays a winning choice if she has one, else she plays a non-losing choice. 
Let us make some observations
(see the pictures in Figure 1):  
\lorip{(For a more detailed justification of these
claims, see the technical report \cite{LORI-VI-techrep}.)}
\begin{enumerate}[leftmargin=*]
\item  NL and SW do not imply each other and neither of them follows from FIR. This can be seen by the following examples.
\smallskip
\begin{itemize}[leftmargin=*]
\item The game $G(1\times 1 + \overline{1\times 1})$ is NL-solvable but not SW-solvable. This is because neither of the players has a winning choice, but if each of them chooses their non-losing move they win.  
%
\item The game $G(Z_2)$ is SW-solvable but not NL-solvable. 
\techrep{
This is because both players have a winning choice, but there are no losing choices.
Note that in this game both players can force winning and thus they both would be sure of winning even without knowing that the other player follows SW.
}
\end{itemize}
\smallskip
\item FIR-solvable games are solvable by \emph{both} SW and NL.
\techrep{
This is because in FIR-solvable games, at least one player $i$ has a winning choice and all the other choices of that player are losing. Hence by following either SW or NL, the player $i$ will select a winning choice. 
}
\smallskip
\item Every BIR-solvable game is \emph{either} NL or SW-solvable.
\techrep{
This is because a BIR-solvable game $G$ is won when every player selects a winning choice, if they have one, or else if they each play a non-losing choice. If at least one player has a winning choice in $G$, then it is SW-solvable, else it is NL-solvable. 
}
\end{enumerate}
\techrep{
\begin{figure}[h]
\begin{center}
\begin{tikzpicture}[scale=0.45,choice/.style={draw, circle, fill=black!100, inner sep=2.1pt}]
	\node at (1,2) {\small\bf $G(1\times 2 + 1\times 0)$};	
	\node at (0,0) [choice] (00) {};
	\node at (2,0) [choice] (20) {};
	\node at (0,1) [choice] (01) {};	
	\node at (2,1) [choice] (21) {};
	\draw[thick] (01) to (21);
	\draw[thick] (01) to (20);	
\end{tikzpicture}
\qquad
\begin{tikzpicture}[scale=0.45,choice/.style={draw, circle, fill=black!100, inner sep=2.1pt}]
	\node at (1,2) {\small\bf $G(1\times 1 + \overline{1\times 1})$};	
	\node at (0,0) [choice] (00) {};
	\node at (0,1) [choice] (01) {};
	\node at (2,0) [choice] (20) {};
	\node at (2,1) [choice] (21) {};
	\draw[thick] (01) to (21);
\end{tikzpicture}
\qquad\quad
\begin{tikzpicture}[scale=0.45,choice/.style={draw, circle, fill=black!100, inner sep=2.1pt}]
	\node at (1,2) {\small\bf $G(Z_2)$};	
	\node at (0,0) [choice] (00) {};
	\node at (2,0) [choice] (20) {};
	\node at (0,1) [choice] (01) {};	
	\node at (2,1) [choice] (21) {};
	\draw[thick] (01) to (21);
	\draw[thick] (00) to (20);	
	\draw[thick] (00) to (21);
\end{tikzpicture}
\label{fig1}
\caption{Some BIR-solvable games}
\smallskip
\end{center}
\end{figure}
}
Therefore, the sets of games solvable by FIR, NL, SW, BIR form the following lattice:

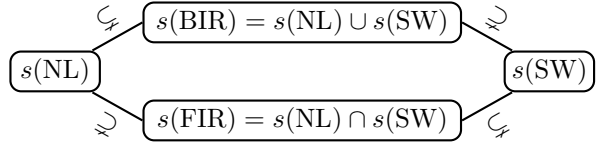
\begin{center}
\begin{tikzpicture}
	[scale=0.8,place/.style={text width=4.7cm, align=flush center, rounded corners, font=\small,
	rectangle,draw=black!100,fill=black!0,thick,inner sep=2.5pt,minimum size=4mm}]
	\node at (0,0) [place] (FIR) {$s(\text{FIR})=s(\text{NL})\cap s(\text{SW})$};
	\node at (5,1) [place, text width=1cm] (SW) {$s(\text{SW})$};
	\node at (-5,1) [place, text width=1cm] (NL) {$s(\text{NL})$};	
	\node at (0,2) [place] (BIR) {$s(\text{BIR})=s(\text{NL})\cup s(\text{SW})$};
	\draw[thick] (FIR.east) to node [auto, sloped, below] {\small $\subsetneq$} (SW);
	\draw[thick] (NL) to node [auto, sloped, below] {\small $\supsetneq$} (FIR.west);
	\draw[thick] (NL) to node [auto, sloped, above] {\small $\subsetneq$} (BIR.west);
	\draw[thick] (BIR.east) to node [auto, sloped, above] {\small $\supsetneq$} (SW);
\end{tikzpicture}
\end{center}


\noindent
SW-solvable and NL-solvable games have simple descriptions. 
In SW-solvable games, at least one player has a surely winning choice. In NL-solvable games, the winning relation forms a nonempty \emph{Cartesian product} between all non-losing choices. BIR-solvable games have (at least) one of these two properties.

\techrep{
Note that in order to follow BIR, the players do not have to make any assumptions on  the behavior or rationality  of each other. In fact, the players do not even need to know that everyone has a mutual goal in the game; that is, following BIR would be equally rational even in coordination games that are \emph{not cooperative}.
}

\subsection{Common beliefs in rationality and iterated reasoning}\label{sec: Common beliefs in rationality}

In contrast to individual rationality, the collective rationality allows players to make assumptions on each other's rationality. Let P be a (purely rational) principle. When \emph{all players believe that everyone follows} P,  
they can reason as follows:
\begin{enumerate}[leftmargin=7mm]
\item[($\star$)] Suppose that by following P each player $i$ must play a choice from $A_i\subseteq C_i$ (i.e., $A_i$ is the smallest set such that $\Sigma(G,i)\subseteq A_i$ for every $\Sigma\in\mathrm{P}$). By this assumption, the players may collectively assume that the actually played game\footnote{In fact, the actually played game may be a proper subgame of $G'$ as the players may also follow other principles.} is $G':=G\upharpoonright (A_1,\dots, A_n)$, and therefore all P-compliant protocols should only prescribe choices in $G'$.  
\end{enumerate}
If players have \emph{common belief} in P being followed, then
the reasoning $(\star)$ above can be repeated for the game $G'$ and this iteration can be continued until a fixed point is reached. 
By $\cir(\mathrm{P})$ we denote the principle of \defstyle{collective iterated reasoning of P}
which prescribes that P is followed 
in the reduced game obtained by the iterated reasoning of ($\star$). 
\lorip{Since every iteration of $(\star)$ only reduces the players' sets of acceptable choices (yet, keeps them nonempty), it is easy to see that $s(\mathrm{P})\subseteq  s(\cir(\mathrm{P}))$ for any principle P (see \cite{LORI-VI-techrep} for more details.)
}
\techrep{Note that after every iteration of ($\star$), the sets of choices for each player become smaller (or remain the same). And since each protocol in any principle P must give nonempty set of choices for any \WLC game, $\cir(\mathrm{P})$ cannot make the set of choices empty for any player. From these observations it is easy to see that $s(\mathrm{P})\subseteq  s(\cir(\mathrm{P}))$ for any principle P.
}

When considering principles of \emph{collective} rationality further, we will apply collective iterated reasoning of the type described above. It may be argued whether such reasoning counts as purely rational, so a question arises: if P is a purely rational principle, is 
$\cir(\mathrm{P})$ 
always purely rational as well? We will not discuss this issue here. We note, however, the extensive literature relating common beliefs and knowledge with individual and collective rationality, see e.g., \cite{FudenbergTirole91}, \cite{Lewis69},  \cite{OR}, \cite{Syverson2002}. See also the following remark on alternative approaches on common belief on a principle P.
%

\begin{remark}\label{rem: iteration order}
It is important to note that when applying $(\star)$, we \emph{simultaneously} remove all the choices (of all players) that are not admissible by $\mathrm{P}$,  and thus the result of the repeated elimination process is unique. However, if we instead eliminated choices of one player at a time, the result of the iterated elimination process could depend on the order in which the players are considered. As we will show in Remark~\ref{rem: eliminating weakly dominated choices}, certain iteration orders may lead to non-rational results. 
Hence one could argue that simultaneous elimination of choices is indeed the most rational way of reasoning when having common belief in a principle~$\mathrm{P}$.
\end{remark}

Related to collective rationality, consider a situation where player $A$ follows certain principles, say $P_1$ and $P_2$, but (s)he is not sure which principles the other players follow. It is now conceivable that $A$ has a strong reason to believe that all the other players follow $P_1$, but $A$ is not sure if these other players also follow $P_2$. 
This could be because $P_1$ seems rationally obvious while $P_2$ is only well justified but not followed by all rational players. Alternatively, $P_2$ could be a more complex principle and $A$ could be skeptical whether the other players are smart enough to follow it.
Consider a scenario (for example a specific class of games) where all players reason in the same way so that they have common belief only in $P_1$ but they in fact all follow both $P_1$ and $P_2$ (but no other principles). The games that the players can now solve correspond to games solvable with the principle $\cir(P_1)\cap P_2$.

\subsection{Basic collective rationality}

Here we extend individually rational principles of Section~\ref{ssec: BIR} by adding common belief in the principles (as described in Section~\ref{sec: Common beliefs in rationality}) to the picture.
We first analyse what happens with the principles NL and SW.
It is easy to see that the collective iterated reasoning of NL reaches a fixed point in a single step by simply removing the losing choices of every player. Hence $s(\mathrm{NL})= s(\cir(\mathrm{NL}))$.
%
Collective iterated reasoning with SW also reaches a fixed point in a single step by eliminating all non-winning choices of every player who has a winning choice. But if even one player has a winning choice, then the game is already SW-solvable. Therefore $s(\mathrm{SW}) = s(\cir(\mathrm{SW}))$.

However, even though common belief in NL or SW does not make them stronger by solving more games, there is a difference on the \emph{epistemic level}. For example, the game $G(1\times 1+\overline{1\times 1})$ is solvable with NL even without common belief in NL. But if both players believe that the other player will follow NL, then they will not only win the game, but they also \emph{believe that that game will be won} before it is played.

Assuming common belief in BIR, some games
which are not BIR-solvable may become solvable. See the following example. 
\begin{example}\label{ex: BCR vs. BIR}
The game $G(Z_2+\overline{1\times 1})$ cannot be solved with $\mathrm{NL}$ or $\mathrm{SW}$. However, if the players can assume that neither of them selects a losing choice (by $\mathrm{NL}$) and eliminate those choices from the game, then they (both) have a winning choice in the reduced game and can win in it by $\mathrm{SW}$.
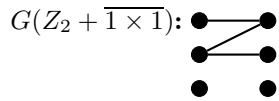
\techrep{
\begin{center}
\begin{tikzpicture}[scale=0.5,choice/.style={draw, circle, fill=black!100, inner sep=2.1pt}]
	\node at (-3.2,2) {\small\bf $G(Z_2+\overline{1\times 1})$: \ \ \ \ };	
	\node at (0,0) [choice] (00) {};
	\node at (2,0) [choice] (20) {};
	\node at (0,1) [choice] (01) {};	
	\node at (2,1) [choice] (21) {};
	\node at (0,2) [choice] (02) {};	
	\node at (2,2) [choice] (22) {};
	\draw[thick] (01) to (21);
	\draw[thick] (02) to (22);
	\draw[thick] (01) to (22);
\end{tikzpicture}
\end{center}
}
\end{example}
Thus, we define the following principle:

\medskip

\textbf{Basic collective rationality (BCR)}:  \; $\cir(\mathrm{BIR})$.

\medskip

\noindent
The example above shows that $s(\mathrm{BIR})\subsetneq s(\mathrm{BCR})$, i.e., BCR is \emph{stronger} than BIR. 
The games solvable by BCR have the following characterisation: 
\textit{after removing all surely losing choices of every player, at least one of the players has a surely winning choice.}

It is worth noting that common belief in SW is not needed for solving games with BCR because a \emph{single} iteration of $\cir(\mathrm{NL})$ suffices. Thus, players could solve BCR-solvable games simply by 
believing everyone to follow NL, i.e., eliminating all losing moves, and then following SW.
By this observation, we have $s(\mathrm{BCR})=s(\cir(\mathrm{NL})+\mathrm{SW})$.
We also point out that the principle BCR is equivalent to the
principle applied in \cite{hawke2017logic} for Strategic Coordination Logic.

\subsection{Principles using optimal choices} 

If a rational player has optimal choices (i.e., at least as good as all other choices), it is natural to assume that she selects such a choice.
\techrep{Note that players may have several optimal choices, or none at all. For example, in the game $G(2\times 2)$ both players have two optimal choices while in $G(Z_3)$ neither of the players has optimal choices.
We now introduce the following principle:
}

\medskip

\textbf{Individual optimal choices (IOC):} 
\textsf{Play an optimal choice, if possible.}


\begin{example}
\label{ex:3.2}
Recall the \WLC game $G^*$ from Example~\ref{ex: 3-player game}. For Casper (who is carrying the crowbar) it is a better choice to go to the front door than to the basement. Likewise, for Jonathan (who is carrying the lantern) it is a better choice to go to the basement than to the front door. Therefore the choice $a_1$ is (the only) optimal choice for player 1 and $b_3$ is (the only) optimal choice for the player 3. 
The player 2 (Jesper) does not have any optimal choices, but if both 1 and 3 play their optimal choices, then the game is won regardless of the choice of 2. Therefore, the game $G^*$ is solvable with $\mathrm{IOC}$. But, since no player has winning or losing choices in this game, it is easy to see that it is not $\mathrm{BCR}$-solvable.
\end{example}

\lorip{
By the description of BIR-solvable games, it is easy to see that they are IOC-solvable. 
We will show that IOC is \emph{incomparable} with BCR (in terms of their sets of solvable games).
As explained above, the game $G^*$ is IOC-solvable but not BCR-solvable. Furthermore, the following BCR-solvable game $G_\Sigma$ is not IOC-solvable since player 1 does not have any optimal choices and so might end up playing a losing choice.

\vspace{-3mm}

\begin{center}
\begin{tikzpicture}[scale=0.45,choice/.style={draw, circle, fill=black!100, inner sep=2.1pt}]
	\node at (-2.5,2)  () {$G_{\Sigma}$:};
	\node at (0,0) [choice] (00) {};
	\node at (2,0) [choice] (20) {};
	\node at (0,1) [choice] (01) {};
	\node at (2,1) [choice] (21) {};
	\node at (0,2) [choice] (02) {};
	\node at (2,2) [choice] (22) {};	
	\draw[thick] (02) to (22);
	\draw[thick] (02) to (21);	
	\draw[thick] (00) to (20);
	\draw[thick] (00) to (21);	
	\node at (-0.7,0) {\small\bf $c_1$};
	\node at (-0.7,1) {\small\bf $b_1$};
	\node at (-0.7,2) {\small\bf $a_1$};
	\node at (2.7,0) {\small\bf $c_2$};
	\node at (2.7,1) {\small\bf $b_2$};
	\node at (2.7,2) {\small\bf $a_2$};		
\end{tikzpicture}
\end{center}
}

\vspace{-1mm}

\techrep{
Note that if a player has winning choices, then the set of optimal choices is the set of winning choices, and therefore $\mathrm{IOC}\subseteq\mathrm{SW}$. From the description of BIR-solvable games, we see that they are also IOC-solvable.

The next example, together with Example \ref{ex:3.2}, shows that IOC is \emph{incomparable} with BCR with respect to the classes of games that are solvable by these principles.
\begin{example}\label{ex: BCR vs. IOC}

Consider the following \WLC game $G_{\Sigma}$.
\begin{center}
\begin{tikzpicture}[scale=0.5,choice/.style={draw, circle, fill=black!100, inner sep=2.1pt}]
	\node at (-2.5,2)  () {$G_{\Sigma}$:};
	\node at (0,0) [choice] (00) {};
	\node at (2,0) [choice] (20) {};
	\node at (0,1) [choice] (01) {};
	\node at (2,1) [choice] (21) {};
	\node at (0,2) [choice] (02) {};
	\node at (2,2) [choice] (22) {};	
	\draw[thick] (02) to (22);
	\draw[thick] (02) to (21);	
	\draw[thick] (00) to (20);
	\draw[thick] (00) to (21);
	\node at (-0.7,0) {\small\bf $c_1$};
	\node at (-0.7,1) {\small\bf $b_1$};
	\node at (-0.7,2) {\small\bf $a_1$};
	\node at (2.7,0) {\small\bf $c_2$};
	\node at (2.7,1) {\small\bf $b_2$};
	\node at (2.7,2) {\small\bf $a_2$};		
\end{tikzpicture}
\end{center}
By following $\mathrm{BCR}$, player 1 chooses either $a_1$ or $c_1$ and player 2 chooses $b_2$, whence the game is won. However, $G_\Sigma$ is not solvable with $\mathrm{IOC}$ since player 1 does not have any optimal choices (and may thus end up choosing the losing choice $b_1$).
\end{example}
}

\lorip{
\vspace{-1mm}

\noindent
In order to avoid pathological cases like this we can add NL to IOC.
}

\techrep{
As we saw earlier, if a player does not have optimal choices, following only IOC might lead to playing a losing choice. 
In order to avoid pathological cases like this, we should at least add NL to IOC.
}

\medskip

\textbf{Improved basic individual rationality (BIR$^+$):} \; IOC $\cap$ NL

\medskip

\lorip{
\noindent
This principle is stronger than BCR (see 
\cite{LORI-VI-techrep}) even though it is based only on \emph{individual} reasoning.
}
%
\noindent
Since IOC $\subseteq$ SW, we have $\mathrm{BIR}^+\subseteq\mathrm{BIR}$. Note that, unlike BCR, the principle BIR$^+$ is only based on \emph{individual} reasoning.
However, BIR$^+$ is nevertheless stronger than BCR as shown by the following proposition.
\begin{proposition}\label{the: BCR vs. IOC+NL}
$s(\mathrm{BCR})\subsetneq s(\mathrm{BIR^+})$.
\end{proposition}
\begin{proof}
Suppose first that $G\in s(\mathrm{BCR})$. Then each player $i$ has a nonempty set $A_i$ of winning choices in the reduced game after removing all losing choices of all the other players. But now every choice in $A_i$ must be an optimal choice of $i$ in the original game $G$. Hence, by following $\mathrm{BIR^+}$, the player $i$ will play a choice from $A_i$ (by IOC) while all the other players play a non-losing choice (by NL), whence the game is won. Therefore $G\in s(\mathrm{BIR^+})$ and thus $s(\mathrm{BCR})\subseteq s(\mathrm{BIR^+})$.
In Example~
\ref{ex:3.2}
we saw that $G^*$ is solvable with IOC but not with BCR. Therefore $s(\mathrm{BCR})\subsetneq s(\mathrm{BIR^+})$.
\end{proof}

\medskip
We now consider the collective version of IOC:

\medskip

\textbf{Collective optimal choices (COC):} \; \cir(IOC)

\medskip

\lorip{
\noindent
We can show that COC is stronger than $\mathrm{BIR^+}$ and therefore also stronger than BCR (see 
\cite{LORI-VI-techrep}). 
Finally, observe that  in a 2-player \WLC game $G$ where $W_G\neq\emptyset$ \emph{the only optimal choices are those that are winning against all non-losing choices of the other player}. Therefore, in the special case of 2-player \WLC games, it is easy to see that the hierarchy collapses as $s(\mathrm{BCR})=s(\mathrm{BIR}^+)=s(\mathrm{COC})$.
}

\techrep{
\begin{proposition}\label{the: BIR+ vs. COC}
$s(\mathrm{BIR^+})\subsetneq s(\mathrm{COC})$.
\end{proposition}

\begin{proof}
We first show that $s(\mathrm{BIR^+})\subseteq s(\mathrm{COC})$. Suppose that a \WLC game $G$ is $\mathrm{BIR^+}$-solvable, i.e., the game is won when every player plays an optimal choice, if they have any, else they play a non-losing choice. Let $G'$ be the game that is obtained after the first collective iteration of \cir(IOC). Now, all the remaining non-losing choices of every player in $G'$ must be winning choices. Since winning choices are also optimal choices, all losing choices are eliminated in the second iteration of \cir(IOC). After that, all combinations of the remaining choices are winning. Thus, the game is won by following COC.

Now, consider the following \WLC game $G^{**}$.
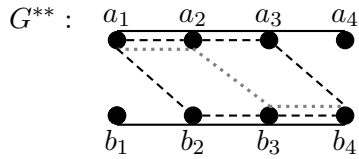
\begin{center}
\begin{tikzpicture}[scale=1.2,choice/.style={draw, circle, fill=black!100, inner sep=2.3pt}]
	\node at (-1,1.3) {$G^{**}:$};
	\node at (0,0) [choice] (00) {};
	\node at (0,0) [choice] (00) {};
	\node at (1,0) [choice] (10) {};
	\node at (2,0) [choice] (20) {};
	\node at (3,0) [choice] (30) {};
	\node at (0,1) [choice] (01) {};
	\node at (1,1) [choice] (11) {};
	\node at (2,1) [choice] (21) {};
	\node at (3,1) [choice] (31) {};
	\draw[thick] (00.south) to (10.south);
	\draw[thick] (10.south) to (20.south);
	\draw[thick, densely dashed] (10) to (20);
	\draw[thick,densely dashed] (01) to (11);
	\draw[thick] (01.north) to (11.north);
	\draw[very thick, dotted, black!50] (01.south) to (11.south);
	\draw[thick] (11.north) to (21.north);
	\draw[thick, densely dashed] (11) to (21);
	\draw[thick, densely dashed] ([yshift=0.1mm]01.south) to (10);
	\draw[very thick, dotted, black!50] ([yshift=0.1mm]11.south) to ([yshift=0.1mm]20.north);
	\draw[thick, densely dashed] (21) to ([yshift=0.1mm]30.north);
	\draw[thick, densely dashed] (20) to (30);
	\draw[very thick, dotted, black!50] (20.north) to (30.north);
	\draw[thick] (20.south) to (30.south);
	\draw[thick] (21.north) to (31.north);
	\node at (0,1.3) {$a_1$};
	\node at (1,1.3) {$a_2$};
	\node at (2,1.3) {$a_3$};
	\node at (3,1.3) {$a_4$};
	\node at (0,-0.35) {$b_1$};
	\node at (1,-0.35) {$b_2$};
	\node at (2,-0.35) {$b_3$};
	\node at (3,-0.35) {$b_4$};
\end{tikzpicture}
\end{center}
Here only players $1$ and $4$ have optimal choices $a_1$ and  $b_4$, respectively, and no player has losing choices. Hence we see that  $G$ cannot be solved with $\mathrm{BIR^+}$.
(By following $\mathrm{BIR^+}$, players may end up selecting the choice profile $(a_1,b_2,a_3,b_4)$ which is not winning.)
However, after the first iteration of \cir(IOC), the players $2$ and $3$ have optimal choices $a_2$ and $b_3$, respectively. Hence, by following COC, the players end up choosing the winning choice profile $(a_1,a_2,b_3,b_4)$.
(Note that we can construct a similar game for $2n$ players, where it takes $n$ iterations of \cir(IOC) for solving it.)
\end{proof}

\medskip
Now, let us consider what happens in the special case of $2$-player \WLC-games. We first observe that \emph{the only optimal choices in a 2-player \WLC game $G$ (where $W_G\neq\emptyset$) are those that are winning against all non-surely losing choices of the other player}. Consequently, when considering 2-player \WLC games, we have $s(\mathrm{IOC})\, \cup \, s(\mathrm{COC})\subseteq s(\mathrm{BCR})$. 
By combining this with the results of Propositions~\ref{the: BCR vs. IOC+NL} and \ref{the: BIR+ vs. COC}, we obtain the following result.

\begin{proposition}\label{the: 2-player case}
For $2$-player \emph{\WLC} games: \; $s(\mathrm{BCR})=s(\mathrm{BIR}^+)=s(\mathrm{COC})$.
\end{proposition}
}

\subsection{Elimination of weakly dominated choices} 

In game-theory, rationality is usually associated with elimination of dominated strategies. 
As noted in Section~\ref{ssec: BIR}, \emph{strict} domination is a too strong
concept for \WLC-games.
Weak domination, on the other hand, gives the following principle when applied individually.
 
\medskip

\textbf{Individually rational choices (IRC)}: 
\textsf{Do not play a choice $a$ when there is a better \\
\indent choice $b$ available. That is, if $W^i_G(a) \subsetneq W^i_G(b)$, then the player $i$ should not play $a$. }

\medskip

\lorip{
\noindent
Note that by this definition, when a player follows IRC, she also follows NL and IOC, 
and therefore $s(\mathrm{BIR^+})\subseteq s(\mathrm{IRC})$. The inclusion is, in fact, proper since the following \WLC game $G_\#$ is solvable with IRC but not with $\mathrm{BIR^+}$.

\vspace{-2mm}

\begin{center}
\begin{tikzpicture}[scale=0.45,choice/.style={draw, circle, fill=black!100, inner sep=2.1pt}]
	\node at (-2.5,2)  () {$G_\#$:};
	\node at (0,0) [choice] (00) {};
	\node at (2,0) [choice] (20) {};
	\node at (0,1) [choice] (01) {};
	\node at (2,1) [choice] (21) {};
	\node at (0,2) [choice] (02) {};
	\node at (2,2) [choice] (22) {};
	\node at (0,-1) [choice] (0-1) {};
	\node at (2,-1) [choice] (2-1) {};
	\draw[thick] (02) to (21);
	\draw[thick] (01) to (22);
	\draw[thick] (01) to (21);
	\draw[thick] (01) to (20);
	\draw[thick] (00) to (20);
	\draw[thick] (00) to (21);	
	\draw[thick] (00) to (2-1);
	\draw[thick] (0-1) to (20);
	\node at (-0.7,0) {\small\bf $c_1$};
	\node at (-0.7,1) {\small\bf $b_1$};
	\node at (-0.7,2) {\small\bf $a_1$};
	\node at (-0.7,-1) {\small\bf $d_1$};
	\node at (2.7,0) {\small\bf $c_2$};
	\node at (2.7,1) {\small\bf $b_2$};
	\node at (2.7,2) {\small\bf $a_2$};
	\node at (2.7,-1) {\small\bf $d_2$};
\end{tikzpicture}
\end{center}
\vspace{-2mm}
}
\techrep{
\noindent
Note that by the definition, $\mathrm{IRC}\subseteq\mathrm{NL}\cap\mathrm{IOC}$ and therefore $s(\mathrm{BIR^+})\subseteq s(\mathrm{IRC})$. The inclusion here is proper since there are \WLC games that are solvable with IRC but not with $\mathrm{BIR^+}$. 
%
\begin{example}\label{ex: BIR+ vs. IRC}
Consider the following \WLC game $G_\#$.
\begin{center}
\begin{tikzpicture}[scale=0.5,choice/.style={draw, circle, fill=black!100, inner sep=2.1pt}]
	\node at (-2.5,2)  () {$G_\#$:};
	\node at (0,0) [choice] (00) {};
	\node at (2,0) [choice] (20) {};
	\node at (0,1) [choice] (01) {};
	\node at (2,1) [choice] (21) {};
	\node at (0,2) [choice] (02) {};
	\node at (2,2) [choice] (22) {};
	\node at (0,-1) [choice] (0-1) {};
	\node at (2,-1) [choice] (2-1) {};
	\draw[thick] (02) to (21);
	\draw[thick] (01) to (22);
	\draw[thick] (01) to (21);
	\draw[thick] (01) to (20);
	\draw[thick] (00) to (20);
	\draw[thick] (00) to (21);	
	\draw[thick] (00) to (2-1);
	\draw[thick] (0-1) to (20);
	\node at (-0.7,0) {\small\bf $c_1$};
	\node at (-0.7,1) {\small\bf $b_1$};
	\node at (-0.7,2) {\small\bf $a_1$};
	\node at (-0.7,-1) {\small\bf $d_1$};
	\node at (2.7,0) {\small\bf $c_2$};
	\node at (2.7,1) {\small\bf $b_2$};
	\node at (2.7,2) {\small\bf $a_2$};
	\node at (2.7,-1) {\small\bf $d_2$};
\end{tikzpicture}
\end{center}
In $G_\#$, all players have neither losing choices nor optimal choices, and therefore it cannot be solved with $\mathrm{BIR^+}$. But the choice $b_1$ is better than $a_1$, and likewise $c_1$ is better than $d_1$. (Note that $b_1$ and $c_1$ are not comparable with each other.) Therefore, by following $\mathrm{IRC}$, the player 1 will play either $b_1$ or $c_1$. With the same reasoning, the player 2 will play either $b_2$ or $c_2$, which leads to a win. Therefore $G^\#\in s(\mathrm{IRC})$.
\end{example}
}
\lorip{
We show in 
\cite{LORI-VI-techrep} that IRC is incomparable with COC. However, based on
the observations above, in the 2-player case $s(\mathrm{COC})=s(\mathrm{BIR^+})\subsetneq s(\mathrm{IRC})$.
}

\techrep{
The COC-solvable game $G^{**}$ (in the proof of Proposition~\ref{the: BIR+ vs. COC}) is not solvable with IRC. 
(This is because neither of the moves $a_2$ and $b_2$ (respectively $a_3$ and $b_3$) is better than the other.)
On the other hand, the game $G_\#$ in Example~\ref{ex: BIR+ vs. IRC} is not solvable with COC, and therefore IRC is incomparable with COC in the general case.
However, in the 2-player case $s(\mathrm{COC})\subsetneq s(\mathrm{IRC})$, since then $s(\mathrm{COC})=s(\mathrm{BIR^+})$ by Proposition~\ref{the: 2-player case}.
}

We next assume common belief in IRC. As commonly known (see e.g.  \cite{OR}), \emph{iterated elimination of weakly dominated strategies}  
eventually stabilises in some reduced game 
but different elimination orders may produce different results (cf. Remark~\ref{rem: eliminating weakly dominated choices}).  
However, when applying \cir(IRC), the process will stabilise to a unique reduced game since all weakly dominated choices are always removed \emph{simultaneously}.
By following the next principle, players will play a choice within this reduced game. 

\medskip

\textbf{Collective  rational choices (CRC)}: \; 
\cir(IRC)
  
\medskip

\lorip{
\noindent
For example, $G(Z_3)$ is not solvable with IRC, but can be solved with CRC by doing two collective iterations of IRC. Thus $s(\mathrm{IRC})\subsetneq s(\mathrm{CRC})$. This observation can be generalized as follows:  to solve a game of the form $G(Z_n)$, the players need $n-1$ iterations of IRC. Therefore different numbers of iterations of IRC form a proper hierarchy of CRC-solvable 2-player \WLC games within $s(\mathrm{CRC})$.
}

\techrep{
\noindent
The following example shows that $s(\mathrm{IRC})\subsetneq s(\mathrm{CRC})$.

\begin{example}\label{ex: Iteration of CRC}
Consider the following \WLC game.
\begin{center}
\begin{tikzpicture}[scale=0.5,choice/.style={draw, circle, fill=black!100, inner sep=2.1pt}]
	\node at (-3,2) {$G(Z_3)$:};
	\node at (0,0) [choice] (00) {};
	\node at (2,0) [choice] (20) {};
	\node at (0,1) [choice] (01) {};
	\node at (2,1) [choice] (21) {};
	\node at (0,2) [choice] (02) {};
	\node at (2,2) [choice] (22) {};	
	\draw[thick] (02) to (22);	
	\draw[thick] (01) to (21);
	\draw[thick] (01) to (22);			
	\draw[thick] (00) to (20);
	\draw[thick] (00) to (21);
	\node at (-0.7,0) {\small\bf $c_1$};
	\node at (-0.7,1) {\small\bf $b_1$};	
	\node at (-0.7,2) {\small\bf $a_1$};		
	\node at (2.7,0) {\small\bf $c_2$};
	\node at (2.7,1) {\small\bf $b_2$};	
	\node at (2.7,2) {\small\bf $a_2$};			
\end{tikzpicture}
\end{center}
We first note that $b_1$ is a better choice than $a_1$ and likewise $b_2$ is a better choice than $c_2$. Therefore, by following $\mathrm{IRC}$, player 1 will play a choice from $\{b_1,c_1\}$ and player 2 will play from $\{a_2,b_2\}$, which does not guarantee winning. However, after eliminating $a_1$ and $c_2$, then $b_1$ is better than $c_1$ and $b_2$ is better than $a_2$. Thus, by following $\mathrm{CRC}$ and doing one more iteration of $\cir(\mathrm{IRC})$, player 1 and 2 have only the choices $b_1$ and $b_2$ which are winning.

In $G(Z_3)$, we needed two iterations of $\cir(\mathrm{IRC})$. It is easy to see that in the game $G(Z_4)$ the iterations are done analogously and the fixed point is reached in 3 iterations. Likewise, we can see that $n-1$ iterations of $\cir(\mathrm{IRC})$ are needed for solving $G(Z_n)$. Therefore, the numbers of iterations of $\cir(\mathrm{IRC})$ form a proper hierarchy of $\mathrm{CRC}$-solvable 2-player \WLC games.
\end{example}
}

\begin{remark}\label{rem: eliminating weakly dominated choices}
As discussed in Remark~\ref{rem: iteration order}, if the iterated process of choice elimination is performed in a
non-simultaneous fashion, the resulting
reduced game might be different depending on the selected elimination order. 
This can easily lead to reasoning patterns that may appear sound, but are arguably irrational.
For example, in the game $G(Z_3)$ above, we could first
eliminate the weakly dominated choice $a_1$ of player 1 and then the choices $a_2$ and $c_2$ of player 2,
thereby ending up with the subgame with the choices $b_1,b_2,c_1$ remaining.
The player 1 could then conclude that $b_1$ and $c_1$ are equally good choices for him.
(Note that neither of these choices is weakly dominated in the original game.)
Symmetrically, we could first eliminate the choice $c_2$ of player 2 and 
then the choices $a_1$ and $c_1$ of player 1, ending up with the game with
the choices $b_1,a_2,b_2$. Now player 2 could conclude that the choices $a_2$ and $b_2$
are equally good for him. This way player $1$ could end up choosing $c_1$ (based on
the first reduced game) and player $2$ choice $a_2$ (based on the second reduced game),
whence the players would not coordinate. 

For an example of even more problematic reasoning leading to irrational results, consider the $\mathrm{SW}$-solvable game $G(Z_2)$. Here player 1 could first assume that player 2 is following $\mathrm{SW}$ and eliminate the non-surely winning choice of player 2. Since the remaining choices of player 1 are surely winning in the reduced game, (s)he could then choose either of them. But similarly, player 2 could assume that player 1 is not playing the non-surely winning choice whence it would be fine (for player 2) to play any choice. As a result, the players could end up choosing a non-winning choice profile and lose $G(Z_2)$.
Since the reasoning here is clearly faulty, this demonstrates that the elimination of choices ``blindly'' in an arbitrary order can lead to  irrational behavior.
\end{remark}

\subsection{Symmetry-based principles} 

By \emph{only} following the concept of rationality from game-theory, one \emph{could} argue that CRC reaches the border of rational principles. However, we now will define more principles which are incomparable with CRC but can still be regarded as purely rational, and hence to be followed by all ideally rational players. 
These principles are based on \emph{symmetries} in \WLC games and the assumption that players follow only structural protocols is central here.

We begin with auxiliary definitions. We say that a choice profile $(c_1,\dots,c_n)$ \defstyle{exhibits a bad choice symmetry} if 
$\llbracket c_{1}\rrbracket_1 \times\cdots\times \llbracket c_{n}\rrbracket_n  \not\subseteq W_G$ 
(recall Definition~\ref{def: choice renaming}),  
and that a choice $c$ \defstyle{generates a bad choice symmetry} if $\sigma_c$ exhibits bad choice symmetry for \emph{every} choice profile $\sigma_c$ that contains $c$.

\medskip
\textbf{Elimination of bad choice symmetries (ECS):} \\
\indent 
\textsf{Never play choices that generate a bad choice symmetry, if possible.}
\medskip

\noindent
Why should this principle be considered rational? Suppose that a player $i$ plays a choice $c_i$ which generates a bad choice symmetry. It is now possible to win only if
some tuple $$(c_1,\dots,c_{i-1},c_i,c_{i+1},\dots,c_n)\in W_G$$ is eventually chosen.
However, due to structural symmetry, the players
have \emph{exactly the same reasons} to play in such a way that any other
tuple in $\llbracket c_{1}\rrbracket_1 \times\cdots\times \llbracket c_{n}\rrbracket_n$
is  selected---and that other tuple may possibly be a losing one since 
$\llbracket c_{1}\rrbracket_1 \times\cdots\times \llbracket c_{n}\rrbracket_n \not\subseteq W_G$.

\lorip{\vspace{-1mm}}
\lorip{\begin{example}
Here is a typical example of using ECS. Suppose that the game graph of $G$ has two (or more) connected components that are isomorphic to each other. Since no player can see a difference between those components, all players should avoid playing choices from them.  
With this reasoning, games like $G(1\times 1 + 2(1\times 2))$ can be solved. Note that this game is not CRC-solvable since no player has any weakly dominated choices.
\end{example}}
\lorip{\vspace{-1mm}}

\techrep{
For a typical example of using ECS, suppose that the game graph of $G$ has two (or more) connected components that are isomorphic to each other. Since no player can detect a difference between those components, all players should avoid playing choices from them.  
%
}

\techrep{
\begin{example} 
Consider the \WLC game $G(1\times 1 + 2(1\times 2))$:
\begin{center}
\begin{tikzpicture}[scale=0.5,choice/.style={draw, circle, fill=black!100, inner sep=2.1pt}]
	\node at (0,0.5) [choice] (00) {};
	\node at (2,0) [choice] (20) {};
	\node at (2,1) [choice] (21) {};
	\node at (0,2.5) [choice] (02) {};
	\node at (2,2) [choice] (22) {};
	\node at (2,3) [choice] (23) {};	
	\node at (0,4) [choice] (04) {};
	\node at (2,4) [choice] (24) {};		
	\draw[thick] (04) to (24);	
	\draw[thick] (02) to (22);
	\draw[thick] (02) to (23);
	\draw[thick] (00) to (20);
	\draw[thick] (00) to (21);	
	\node at (-0.7,0.5) {\small\bf $c_1$};
	\node at (-0.7,2.5) {\small\bf $b_1$};	
	\node at (-0.7,4) {\small\bf $a_1$};		
	\node at (2.7,0) {\small\bf $e_2$};
	\node at (2.7,1) {\small\bf $d_2$};
	\node at (2.7,2) {\small\bf $c_2$};
	\node at (2.7,3) {\small\bf $b_2$};	
	\node at (2.7,4) {\small\bf $a_2$};			
\end{tikzpicture}
\end{center}
In this game $b_1\simeq_1 c_1$ and $b_2\simeq_2 c_2\simeq_2 d_2\simeq_2 e_2$. Since all the choice profiles in $\{b_1,c_1\}\times\{b_2,c_2,d_2,e_2\}$ are not winning, we see that both $b_1$ and $c_1$ generate a bad choice symmetry. Likewise, $b_2$, $c_2$, $d_2$ and $e_2$ generate a bad choice symmetry. Therefore, by following ECS, the players will choose $a_1$ and $a_2$.
\end{example} 
 }

While ECS only considers symmetries between similar choices, the next principle takes symmetries \emph{between players} into account. 
Consider a choice profile $\vec c=(c_1,\dots,c_n)$ and let $S^p_i(\vec c):=\{c_i\}\cup(C_i\cap\bigcup_{j\neq i}[c_j])$ for each $i$ 
(recall Definition \ref{deftv}). 
We say that $(c_1,...,c_n)$ \defstyle{exhibits a bad player symmetry} 
if $S^p_1(\vec c)\times\dots\times S^p_n(\vec c)\not\subseteq W_G$
and a choice $c$ \defstyle{generates a bad player symmetry} if $\sigma_c$ exhibits a bad player symmetry for every choice profile $\sigma_c$ that contains $c$.

\medskip
\textbf{Elimination of bad player symmetries (EPS):} \\ 
\indent
\textsf{Never play choices that generate bad player symmetries, if possible.}
\medskip

\noindent
Here the players assume that all players reason similarly, or alternatively, each player wants to 
play so that she would at least coordinate with herself in the case she was to use her 
protocol to make a choice in each player role of a \WLC game.
Suppose that the players have some reasons to select a choice profile $(c_1,\dots,c_n)$. Now, if there are players $i\neq j$ and  a choice $c_j'\in C_j$ such that $c_j'\sim c_i$, then the player $j$ should have the same reason to play $c_j'$ as $i$ has for playing $c_i$. Hence, if the players have their reasons to play $(c_1,\dots,c_n)$, they should have the same reasons to play any choice profile in $S^p_1(\vec c)\times\dots\times S^p_n(\vec c)$. 
Winning is not guaranteed if $S^p_1(\vec c)\times\dots\times S^p_n(\vec c)\not\subseteq W_G$.

It is worth noting that EPS bears a close resemblance to the notion of \emph{superrationality} defined by Hofstadter \cite{Hofstadter}. 

\begin{example}
Consider $\mathrm{EPS}$ in the case of a two-player game \WLC game $G$. If for a given choice $c\in C_1$, there is a structurally equivalent choice $c'\in C_2$ such that $(c,c')\notin W_G$, then by following $\mathrm{EPS}$, player 1 does not play the choice $c$ (and likewise player 2 does not play the choice $c'$).
With this kind of reasoning, some $\mathrm{CRC}$-unsolvable games like $G(1\times 1 + 1\times 2 + 2\times 1)$ below become solvable.
\begin{center}
\begin{tikzpicture}[scale=0.5,choice/.style={draw, circle, fill=black!100, inner sep=2.1pt}]
	\node at (0,0) [choice] (00) {};
	\node at (2,0) [choice] (20) {};
	\node at (0,1) [choice] (01) {};	
	\node at (2,1) [choice] (21) {};
	\node at (0,2) [choice] (02) {};	
	\node at (2,2) [choice] (22) {};
	\node at (0,-1) [choice] (2-1) {};
	\node at (2,-1) [choice] (2-1) {};
	\draw[thick] (00) to (2-1);
	\draw[thick] (02) to (22);
	\draw[thick] (01) to (21);
	\draw[thick] (01) to (20);
	\draw[thick] (0-1) to (2-1);
	\node at (-0.7,2) {\small $a_1$};
	\node at (-0.7,1) {\small $b_1$};
	\node at (-0.7,0) {\small $c_1$};
	\node at (-0.7,-1) {\small $d_1$};
	\node at (2.7,2) {\small $a_2$};
	\node at (2.7,1) {\small $b_2$};
	\node at (2.7,0) {\small $c_2$};
	\node at (2.7,-1) {\small $d_2$};
\end{tikzpicture}
\end{center}

Note also that the game $G^*$ from Example~\ref{ex: 3-player game} is $\mathrm{EPS}$-solvable since both choices $b_1$ and $a_3$ generate a bad player symmetry.
\end{example}
\lorip{\vspace{-1mm}}

\techrep{
\begin{example}
In Example~\ref{ex: Iteration of CRC} we showed that in order to solve $G(Z_n)$ by $\mathrm{CRC}$ it takes $n-1$ collective iterations and after that the ``middle choices'' are selected by both of the players. The game $G(Z_n)$ can also be solved by $\mathrm{EPS}$ and the players will end up choosing the same choices as with $\mathrm{CRC}$. This is because every other choice---except the middle choice---generates a bad player symmetry. 
\end{example}
}

Finally, we introduce a principle that takes both types of symmetries into account.
For a choice profile $\vec c=(c_1,...,c_n)$ let $S_i(\vec c):=C_i\cap\bigcup_j[c_j]$ for each $i$.
We say that $(c_1,...,c_n)$ \defstyle{exhibits a bad symmetry} 
if $S_1(\vec c)\times\dots\times S_n(\vec c)\not\subseteq W_G$, 
and a choice $c$ \defstyle{generates a bad symmetry} if $\sigma_c$ exhibits a bad symmetry for every
choice profile $\sigma_c$ that contains $c$.

\medskip
\textbf{Elimination of bad symmetries (ES):} \\
\indent
\textsf{Never play choices that generate bad symmetries, if possible.}
\medskip

\lorip{
\noindent
It is easy to show that ECS and EPS are not comparable and that they are both weaker than ES. Furthermore,  all symmetry based principles can clearly solve NL-solvable games, but they are incomparable with SW and all the stronger principles.
For proofs of these claims and further examples, see~\cite{LORI-VI-techrep}.
}

\techrep{
\noindent
It is easy to see from the definition of bad symmetry that if a choice $c$ generates either a bad choice symmetry or a bad player symmetry, then $c$ also generates a bad symmetry. Therefore, by using Claim I from Section~\ref{sec:unsolvable}, it is easy to show that $s(\mathrm{ECS}) \cup s(\mathrm{EPS})\subseteq s(\mathrm{ES})$.

By the definitions of ECS and EPS, it is clear that they can solve all NL-solvable games and therefore also ES can solve all NL-solvable games. Furthermore, we can show that the classes of games solvable by ECS, EPS, and CRC are \emph{completely independent} of each other.
See the table in Figure \ref{fig:table}. 

\medskip

\begin{figure}[h]
\begin{center}
\begin{tabular}{|l|c|}
\hline
Class of games $\mathcal{G}$  & Example of a game in the class $\mathcal{G}$ \\
\hline
$s(\mathrm{ECS})\setminus(s(\mathrm{EPS})\cup s(\mathrm{CRC}))$ & $G(1\times 1 + 2(1\times 2))$ \\
\hline
$s(\mathrm{EPS})\setminus(s(\mathrm{ECS})\cup s(\mathrm{CRC}))$ & $G(1\times 1 + 1\times 2 + 2\times 1)$ \\
\hline
$s(\mathrm{CRC})\setminus(s(\mathrm{ECS}\cup s(\mathrm{EPS}))$ & 
\begin{tikzpicture}[scale=0.45,choice/.style={draw, circle, fill=black!100, inner sep=2.1pt}]
	\node at (2,0) [choice] (20) {};
	\node at (0,1) [choice] (01) {};
	\node at (2,1) [choice] (21) {};
	\node at (0,2) [choice] (02) {};
	\node at (2,2) [choice] (22) {};	
	\node at (0,2.5) {};	
	\node at (0,-0.5) {};	
	\draw[thick] (02) to (22);	
	\draw[thick] (01) to (21);
	\draw[thick] (01) to (22);			
	\draw[thick] (01) to (20);
\end{tikzpicture}
\\
\hline
$(s(\mathrm{ECS})\cap s(\mathrm{EPS}))\setminus s(\mathrm{CRC})$ & $G(1\times 1 + 2(2\times 2))$ \\
\hline
$(s(\mathrm{ECS})\cap s(\mathrm{CRC}))\setminus s(\mathrm{EPS})$ & 

\begin{tikzpicture}[scale=0.45,choice/.style={draw, circle, fill=black!100, inner sep=2.1pt}]
	\node at (0,0.5) [choice] (00) {};
	\node at (2,0) [choice] (20) {};
	\node at (2,1) [choice] (21) {};
	\node at (0,1.5) [choice] (02) {};
	\node at (2,2) [choice] (22) {}; 
	\node at (0,2.5) {};	
	\node at (0,-0.5) {};	
	\draw[thick] (02) to (22);
	\draw[thick] (02) to (21);	
	\draw[thick] (00) to (20);
	\draw[thick] (00) to (21);		
\end{tikzpicture}
 \\
\hline
$(s(\mathrm{EPS})\cap s(\mathrm{CRC}))\setminus s(\mathrm{ECS})$ & $G(Z_3)$ \\
\hline
\end{tabular}
\caption{Mutual independence of the principles ECS, EPS, and CRC.}
\label{fig:table}
\end{center}
\end{figure}

\medskip

The \WLC game given on Figure \ref{fig:table} in the class $s(\mathrm{CRC})\setminus(s(\mathrm{ECS}\cup s(\mathrm{EPS}))$  is also unsolvable with ES and therefore ES and CRC are incomparable with each other. This particular game is also SW-solvable, and thus it follows that all symmetry based principles are incomparable with SW. Since ECS and EPS are incomparable and $s(\mathrm{ECS}),s(\mathrm{EPS})\subseteq s(\mathrm{ES})$, it also follows that ES is stronger than both ECS and EPS.

So far we have only presented examples of such ECS-solvable games that contain isomorphic connected components. In the following example, we see how ECS can be used for eliminating moves from a single component. This particular example can also be solved with EPS (and ES) but not with CRC.

\begin{example}
In the \WLC game $G(O_3 + 1\times 1)$, there are no weakly dominated choices and thus it is not $\mathrm{CRC}$-solvable. However, by applying  \emph{ECS}, \emph{EPS} or \emph{ES}, players will play choises $d_1$ and $d_2$ which are winning.
\begin{center}
\begin{tikzpicture}[scale=0.5,choice/.style={draw, circle, fill=black!100, inner sep=2.1pt}]
	\node at (-4.5,2) {\small $G(O_3 + 1\times 1):$};	
	\node at (0,0) [choice] (00) {};
	\node at (2,0) [choice] (20) {};
	\node at (0,1) [choice] (01) {};	
	\node at (2,1) [choice] (21) {};
	\node at (0,2) [choice] (02) {};	
	\node at (2,2) [choice] (22) {};
	\node at (0,-1) [choice] (2-1) {};
	\node at (2,-1) [choice] (2-1) {};
	\draw[thick] (00) to (20);	
	\draw[thick] (00) to (21);
	\draw[thick] (02) to (22);
	\draw[thick] (01) to (22);
	\draw[thick] (02) to (21);
	\draw[thick] (01) to (20);
	\draw[thick] (0-1) to (2-1);
	\node at (-0.7,2) {\small $a_1$};
	\node at (-0.7,1) {\small $b_1$};
	\node at (-0.7,0) {\small $c_1$};
	\node at (-0.7,-1) {\small $d_1$};
	\node at (2.7,2) {\small $a_2$};
	\node at (2.7,1) {\small $b_2$};
	\node at (2.7,0) {\small $c_2$};
	\node at (2.7,-1) {\small $d_2$};
\end{tikzpicture}
\end{center}
\end{example}

Symmetry principles can solve many games whose game graphs consist of several connected components. It is easy to see that none of these kinds of games can be solved with CRC nor with any other principles we have presented in this paper. 
However, when considering games whose graph is connected, it is not obvious whether symmetry principles could solve any more games than CRC. In the next example we show that, indeed, there are  games, whose graph consists only of a single component, such that symmetry principles solve them, but they are not solvable by CRC (nor any other principles presented so far).

\begin{example}
In the \WLC game $G$ below, there are no weakly dominated choices. However, by applying \emph{ECS}, \emph{EPS} or \emph{ES}, players will pick choices $c_1$ and $c_2$ which are winning. (Note here that $G$ is almost of the type of $G(O_5)$, the only difference being a single extra edge that ``forms a diagonal of the 10-cycle''.)
\begin{center}
\begin{tikzpicture}[scale=0.5,choice/.style={draw, circle, fill=black!100, inner sep=2.1pt}]
	\node at (-2,4) {\small $G:$};
	\node at (0,0) [choice] (00) {};
	\node at (2,0) [choice] (20) {};
	\node at (0,1) [choice] (01) {};
	\node at (2,1) [choice] (21) {};
	\node at (0,2) [choice] (02) {};
	\node at (2,2) [choice] (22) {};
	\node at (0,3) [choice] (03) {};
	\node at (2,3) [choice] (23) {};	
	\node at (0,4) [choice] (04) {};
	\node at (2,4) [choice] (24) {};		
	\draw[thick] (04) to (24);
	\draw[thick] (04) to (23);
	\draw[thick] (03) to (24);
	\draw[thick] (03) to (22);
	\draw[thick] (02) to (23);
	\draw[thick] (02) to (22);
	\draw[thick] (02) to (21);
	\draw[thick] (01) to (22);	
	\draw[thick] (01) to (20);	
	\draw[thick] (00) to (20);
	\draw[thick] (00) to (21);	
	\node at (-0.7,0) {\small\bf $e_1$};
	\node at (-0.7,1) {\small\bf $d_1$};
	\node at (-0.7,2) {\small\bf $c_1$};
	\node at (-0.7,3) {\small\bf $b_1$};	
	\node at (-0.7,4) {\small\bf $a_1$};		
	\node at (2.7,0) {\small\bf $e_2$};
	\node at (2.7,1) {\small\bf $d_2$};
	\node at (2.7,2) {\small\bf $c_2$};
	\node at (2.7,3) {\small\bf $b_2$};	
	\node at (2.7,4) {\small\bf $a_2$};			
\end{tikzpicture}
\end{center} 
\end{example} 
}

\subsection{Compatibility of ES and CRC}

Recall that ES is the strongest of the symmetry principles and CRC is the strongest of all the other principles presented in this paper. Since they are incomparable (with respect to solvable games), it is natural to ask whether they can be combined.
We first prove the following lemma about the relationship between IRC and ES.

\begin{lemma}\label{dominance vs. symmetries}
Let $G=(A,C_1,\dots,C_n,W_G)$ be \WLC game and let $c\in C_i$ for some $i\leq n$ be a choice that does not generate a bad symmetry. Then none of the choices $d\in C_i$ which are better than $c$ can generate a bad symmetry.
\end{lemma}

\begin{proof}
For the sake of contradiction, assume that there is a choice $d\in C_i$ which generates bad symmetry and which is better than $c$. We first observe that the following holds (by the definition of structural equivalence).
\begin{align*}
	\text{If $d\sim d'$ for some $d'\in C_j$ ($j\leq n$), then there is a choice $c'\in C_j$} \tag{$\star$} \\
	\text{such that $c\sim c'$ and $d'$ is better $c'$.}
\end{align*}

Because $c$ does not generate a bad symmetry, there is a choice profile $\vec c$ of the form $(c_1,\dots, c, \dots c_n)$ which does not exhibit a bad symmetry, i.e. $S_1(\vec c\,)\times\cdots\times S_n(\vec c\,)\subseteq W_G$. In particular $\vec c\in W_G$. Let $\vec d:=(c_1,\dots,d,\dots,c_n)$. Since $d$ is better than $c$, we have $\vec d\in W_G$. Because $d$ generates a bad symmetry, $\vec d$ must exhibit a bad symmetry. Therefore there is a choice profile $\vec e\in S_1(\vec d\,)\times\cdots\times S_n(\vec d\,)$ such that $\vec e\notin W_G$.

Now $\vec e$ must contain at least one choice $d'$, for which $d\sim d'$, as otherwise we would have $\vec e\in S_1(\vec c\,)\times\cdots\times S_n(\vec c\,)\subseteq W_G$. Now, by $(\star)$, for every such choice $d'$, there is a choice $c'$ such that $c\sim c'$ and $d'$ is better than $c'$. Let $\vec e\,[c'/d']$ be a choice profile which is obtained from $\vec e$ by replacing every $d'$, for which $d\sim d'$, with a corresponding choice $c'$. Since $\vec e\notin W_G$ and all choices $d'$ are better than the corresponding choice $c'$, we must have $\vec e\,[c'/d']\notin W_G$. But this is a contradiction since $\vec e\,[c'/d']\in S_1(\vec c\,)\times\cdots\times S_n(\vec c\,)\subseteq W_G$.
\end{proof}

\medskip

We are now ready to show that ES and CRC can be combined simply by taking their intersection.

\begin{proposition}
In every \WLC game it is possible to follow both the principle $\mathrm{ES}$ and the principle $\mathrm{CRC}$, i.e., $\mathrm{CRC}\cap\mathrm{ES}\neq\emptyset$.
\end{proposition}

\begin{proof}
Let $G$ be a \WLC game and let $i$ be a player in $G$. We first note that if all choices of $i$ generate a bad symmetry, then the principle ES does not set any restrictions and thus, by following CRC, the player $i$ also trivially follows $\mathrm{CRC}\cap\mathrm{ES}$.

Suppose then that $i$ has at least one choice $c$ which does not generate a bad symmetry. Now, by Lemma~\ref{dominance vs. symmetries}, there cannot be any choice $d$ which is better than $c$ such that $d$ generates a bad symmetry. Therefore it follows, by induction on the steps of the choice elimination process, that the iteration of IRC cannot remove all choices of $i$ which do not generate bad symmetry. Hence in the reduced game, obtained by CRC, there must still be at least one choice which does not generate a bad symmetry. Now $i$ will follow $\mathrm{CRC}\cap\mathrm{ES}$ by selecting any such choice.
\end{proof}

\medskip

As shown by the following example, $s(\mathrm{CRC})\cup s(\mathrm{ES}) \subsetneq s(\mathrm{CRC}\cap\mathrm{ES})$. And
therefore $\mathrm{CRC}\cap\mathrm{ES}$ is the strongest of all the principles we have presented so far.

\begin{example}
Consider the following \WLC game:
\begin{center}
\begin{tikzpicture}[scale=0.5,choice/.style={draw, circle, fill=black!100, inner sep=2.1pt}]
	\node at (-2.5,4) {$G^\dagger:$};
	\node at (2,0) [choice] (20) {};
	\node at (0,1) [choice] (01) {};
	\node at (2,1) [choice] (21) {};
	\node at (0,2) [choice] (02) {};
	\node at (2,2) [choice] (22) {};	
	\node at (0,3) [choice] (03) {};
	\node at (2,3) [choice] (23) {};	
	\node at (0,4) [choice] (04) {};
	\node at (2,4) [choice] (24) {};	
	\draw[thick] (04) to (24);
	\draw[thick] (03) to (23);
	\draw[thick] (02) to (22);	
	\draw[thick] (01) to (21);
	\draw[thick] (01) to (22);			
	\draw[thick] (01) to (20);
%
	\node at (-0.7,1) {\small\bf $d_1$};	
	\node at (-0.7,2) {\small\bf $c_1$};
	\node at (-0.7,3) {\small\bf $b_1$};	
	\node at (-0.7,4) {\small\bf $a_1$};		
	\node at (2.7,0) {\small\bf $e_2$};	
	\node at (2.7,1) {\small\bf $d_2$};	
	\node at (2.7,2) {\small\bf $c_2$};
	\node at (2.7,3) {\small\bf $b_2$};	
	\node at (2.7,4) {\small\bf $a_2$};		
\end{tikzpicture}
\end{center}
By following $\mathrm{CRC}$ in $G^\dagger$, the player 1 will not play the choice $c_1$ and the player 2 will not play $d_2$ or $e_2$. And by following $\mathrm{ES}$, the player 1 will play neither the choice $a_1$ nor $b_1$ and likewise the player 2 will play neither the choice $a_2$ nor $b_2$. Therefore $G^\dagger$ cannot be surely won by following either $\mathrm{CRC}$ or $\mathrm{ES}$ alone. But by following $\mathrm{CRC}\cap\mathrm{ES}$, the player 1 will select $d_1$ and the player 2 will select $c_2$, whence the game is surely won.
\end{example}

It is important to note that, by following $\mathrm{CRC}\cap\mathrm{ES}$, the principle ES is applied to the \emph{original} game and not to the reduced game that is obtained by iterating IRC. If we applied ES to the reduced game instead, this would lead to to a different kind of principle that allows (arguably) irrational choices. See the following example.

\begin{example}
Consider the following \WLC game:
\begin{center}
\begin{tikzpicture}[scale=0.5,choice/.style={draw, circle, fill=black!100, inner sep=2.1pt}]
	\node at (-5.5,4) {$G(2(1\times 1)+Z_3):$};
	\node at (0,0) [choice] (00) {};
	\node at (2,0) [choice] (20) {};
	\node at (0,1) [choice] (01) {};
	\node at (2,1) [choice] (21) {};
	\node at (0,2) [choice] (02) {};
	\node at (2,2) [choice] (22) {};	
	\node at (0,3) [choice] (03) {};
	\node at (2,3) [choice] (23) {};	
	\node at (0,4) [choice] (04) {};
	\node at (2,4) [choice] (24) {};	
	\draw[thick] (04) to (24);
	\draw[thick] (03) to (23);
	\draw[thick] (02) to (22);	
	\draw[thick] (01) to (21);
	\draw[thick] (01) to (22);	
	\draw[thick] (00) to (21);			
	\draw[thick] (00) to (20);
	\node at (-0.7,0) {\small\bf $e_1$};	
	\node at (-0.7,1) {\small\bf $d_1$};	
	\node at (-0.7,2) {\small\bf $c_1$};
	\node at (-0.7,3) {\small\bf $b_1$};	
	\node at (-0.7,4) {\small\bf $a_1$};		
	\node at (2.7,0) {\small\bf $e_2$};	
	\node at (2.7,1) {\small\bf $d_2$};	
	\node at (2.7,2) {\small\bf $c_2$};
	\node at (2.7,3) {\small\bf $b_2$};	
	\node at (2.7,4) {\small\bf $a_2$};		
\end{tikzpicture}
\end{center}
By iterating $\mathrm{IRC}$ on $G(2(1\times 1)+Z_3)$, the choices $c_1$, $c_2$, $e_1$ and $e_2$ are eliminated and we obtain the following reduced game:
\begin{center}
\begin{tikzpicture}[scale=0.5,choice/.style={draw, circle, fill=black!100, inner sep=2.1pt}]
	\node at (0,0) [choice] (00) {};
	\node at (2,0) [choice] (20) {};
	\node at (0,1) [choice] (01) {};
	\node at (2,1) [choice] (21) {};
	\node at (0,2) [choice] (02) {};
	\node at (2,2) [choice] (22) {};
	\draw[thick] (00) to (20);
	\draw[thick] (01) to (21);
	\draw[thick] (02) to (22);	
	\node at (-0.7,0) {\small\bf $d_1$};
	\node at (-0.7,1) {\small\bf $b_1$};	
	\node at (-0.7,2) {\small\bf $a_1$};		
	\node at (2.7,0) {\small\bf $d_2$};
	\node at (2.7,1) {\small\bf $b_2$};	
	\node at (2.7,2) {\small\bf $a_2$};		
\end{tikzpicture}
\end{center}

Since all the choices in this reduced game generate a bad symmetry, the game is not solvable by $\mathrm{ES}$.
However, by following $\mathrm{CRC}\cap\mathrm{ES}$ in $G(2(1\times 1)+Z_3)$, the players will select the choice profile $(d_1,d_2)$ and win. (In this particular example, the same result is obtained by following $\mathrm{ES}$ alone as all the other choice profiles exhibit a bad symmetry.)

Notice that the choice profile $(d_1,d_2)$ does not exhibit a bad symmetry in the original game, but it exhibits a bad symmetry in the reduced game after the iteration of $\mathrm{IRC}$. By this example, we claim that it is questionable to apply $\mathrm{ES}$ in any reduced game that is obtained by eliminating choices that are excluded by other principles. 
\end{example}

\section{On the limits of pure rationality} 

How far can we go up the hierarchy of purely rational principles? This seems a genuinely difficult question.
In this section we will first study the compatibility of purely rational principles and then present two simple principles which are naturally justified, but go beyond pure rationality. Finally we will present a complete hierarchy (with respect to solvable games) of all principles defined in this paper and discuss whether one can identify a strongest purely rational principle.

\subsection{Merging purely rational principles}\label{sec: Merging purely rational principles}

Since we have defined that purely rational principles are followed by all ideally rational players, it is clear that purely rational principles are closed under intersections. That is, if $\mathrm{P}_1$ and $\mathrm{P_2}$ are purely rational, then $\mathrm{P}_1\cap\mathrm{P_2}$ is purely rational, too. This gives us a method to identify which principles cannot be purely rational: if $\mathrm{P_1}$ or $\mathrm{P_2}$ are \defstyle{incompatible}, in sense that $\mathrm{P}_1\cap\mathrm{P_2}=\emptyset$,
then it must be that at least one of $\mathrm{P_1}$ and $\mathrm{P_2}$ is not purely rational.
However, even if the (extensions of) $\mathrm{P}_1$ and $\mathrm{P}_2$ are incompatible in this way, sometimes it may be that the intensions (recall Remark \ref{rem: Intensions}) behind both $\mathrm{P}_1$ and $\mathrm{P}_2$ seem rational. See the following example.

\begin{example}
For each each positive integer $k$, consider the following principle $\mathrm{P}_k$\hspace{0.5mm}:

\medskip 

\centerline{
If the game contains two or more copies of $G(k\times k)$ as disjoint components,
}
\centerline{
then do not choose from these components, if possible.
}

\medskip

\noindent
This may seem a purely rational principle for all $k$, as it is obvious that if
there are two or more disjoint components $k\times k$ in a game,
then choosing from such a component means that a win cannot by
guaranteed due to the same symmetry-based 
idea as the one used in the justification of $\mathrm{ECS}$.
However, for example $\mathrm{P}_2$ and $\mathrm{P}_3$ are not compatible, as
for example in the game $G(2(2\times 2) + 2(3\times 3))$, following 
both $\mathrm{P}_2$ and $\mathrm{P}_3$ would
leave the players with no choices at all. 
%
%
%
(Note, however, that the symmetry principle $\mathrm{ECS}$ 
deals with all $\mathrm{P}_k$, for all $k$, simultaneously).
\end{example}

Even if the extensions of the principles $\mathrm{P_1}$ and $\mathrm{P_2}$ are incompatible, it is still possible that the informal background \emph{intuition} of each of their \emph{intensions} is rational and thus $\mathrm{P_1}$ and $\mathrm{P_2}$ can be combined in some natural way by taking both of the intensions into account. Suppose, e.g., that the intension of $\mathrm{P_1}$ is ``do not play choices of type A, if possible'' and the intension of $\mathrm{P_2}$ is ``do not play choices of type B, if possible'' (as in the example above).
Now we can naturally combine these intensions into a new intension: ``do not play choices of type A or B, if possible''. The resulting principle will now prescribe \emph{any choice}  in those games where all choices are ``bad'' due to being either of type A or B.

From these observations we point out that we may not be sure that (the extension of) a princinciple P is purely rational unless we know that it is compatible with all purely rational principles. Nevertheless, we may still know that the intuition behind the intension of P is rationally justified.

\subsection{Probabilistic reasoning vs. Occam razor} 

We now mention---without providing precise formal definitions---two structural principles for which it may  seem somewhat controversial to claim them 
purely rational in our sense,
but they are definitely meaningful and natural nevertheless.

The first one is the \defstyle{principle of probabilistically optimal reasoning (PR)}.
Informally put, this principle prescribes to always play a choice that
has as large winning extension as possible. 
Such choices have the highest probability of winning, supposing that all the other players play randomly (but \emph{not} if the others follow PR, too: consider e.g. $G(1\times 2 + 2\times 1)$).  Note, however, that following PR can violate the symmetry principles, as demonstrated by the game $G(2(2\times 2) + 1\times 1)$. So, its application should be suitably restricted, as part of its formulation, to be only applied when none of the so far discussed principles of pure rationality apply.


With PR one can solve games like $G(1\times 1 + 2\times 2)$ that are unsolvable with all other principles presented  here. 
However, in $G(1\times 1 + 2\times 2)$ one could also reason (perhaps less convincingly) that both players should pick their choices from the subgame $G(1\times 1)$ since that is the `\emph{simplest}' (and, also the only `\emph{unique}') winning choice profile.
We call this kind of reasoning the \defstyle{Occam razor principle (OR)}. 
Its intuitive and inevitably imprecise formulation is ``\textsf{If there is one simplest way to coordinate, make that choice.}'' This principle relates to the idea of \emph{focal points} \cite{mehta1994focal}, \cite{Schelling60}, \cite{sugden1995theory}. 
%
Different technically precise formulations are possible, but it is not easy to see which one of them (if any) would best capture the spirit of the principle. 

Note that $G(1\times 1 + 2\times 2)$ can be won if both players follow PR or if both follow OR, but not if one follows PR while the other follows OR. 
Moreover, in this game it is impossible for a player to follow \emph{both} PR and OR. Hence, \emph{at least one} of these principles is not purely rational. 
Actually, it can be argued that \emph{neither of them} is purely rational. 
%
%

\begin{remark}
Consider a setting where players have good reasons to assume that all the other players are playing randomly. (One could argue that they can make this assumption ``by default'' if there is no common belief in rationality.) Now, if a player does not have a winning choice, then winning is not guaranteed. Therefore it makes sense to ``optimize the expectation of winning'' and therefore follow the principle $\mathrm{PR}$. If all the players happen to reason this way, then games like $G(1\times 1 + 2\times 2)$ are in fact won.

However, if we assume common belief in rationality, then it seems obvious that the players can no longer assume that the other players play randomly. Therefore the justification for following $\mathrm{PR}$ becomes questionable, and thus, it is no more clear whether following $\mathrm{PR}$ is rational in games like $G(1\times 1 + 2\times 2)$. One could even argue, based on this observation, that \emph{common belief in rationality can sometimes be harmful for the players}.
This is because the use of $\mathrm{PR}$ was based on the assumption that the other players are playing randomly, and this is no longer the case if they are assumed to be rational. 
\end{remark}


\subsection{Hierarchy of rationality principles with respect to solvable games} 

The partially ordered diagram in Figure~\ref{fig:table2} presents the hierarchy of solvable games with the rationality principles that we have presented in this paper. The principles that only use individual reasoning are put in normal (single) frames and the ones that use collective reasoning have double frames. 
%
Note that the diagram in Figure~\ref{fig:table2} is complete in the sense that no new lines can be added to it (neither in the general nor in the 2-player case).


\begin{figure}[h]
\begin{center}
\begin{tikzpicture}
	[scale=0.6,place/.style={align=flush center, rounded corners, font=\footnotesize,
	rectangle,draw=black!100,fill=black!0,thick,inner sep=3pt,minimum size=4mm}]
	\node [rotate=90] at (4.5,6.5) {\Large $\subsetneq$};
	\node at (0,1) [place] (FIR) {$s(\mathrm{FIR})$};
	\node at (4.0,2) [place] (SW) {$s(\mathrm{SW})$};
	\node at (-4.0,2) [place] (NL) {$s(\mathrm{NL})$};	
	\node at (0,3) [place] (BIR) {$s(\mathrm{BIR})$};
	\node at (0,5) [place,double] (BCR) {$s(\mathrm{BCR})$};
	\node at (4,4) [place] (IOC) {$s(\mathrm{IOC})$};
	\node at (0,7) [place] (BIR+) {$s(\mathrm{BIR}^+)$};
	\node at (3.2,8.25) [place,double] (COC) {$s(\mathrm{COC})$};
	\node at (-3.2,9.25) [place] (IRC) {$s(\mathrm{IRC})$};
	\node at (0,10.5) [place,double] (CRC) {$s(\mathrm{CRC})$};
	\node at (-1,12.5) [place,double] (CRC+ES) {$s(\mathrm{CRC\cap ES})$};
	\node at (-6.4,5) [place] (ELS) {$s(\mathrm{ECS})$};
	\node at (-3.7,5) [place] (ELR) {$s(\mathrm{EPS})$};
	\node at (-6.4,9.25) [place] (ES) {$s(\mathrm{ES})$};
	\draw[thick] (FIR.east) to (SW);
	\draw[thick] (NL) to (FIR.west);
	\draw[thick] (NL) to (BIR.west);
	\draw[thick] (BIR.east) to (SW);
	\draw[thick] (BIR) to (IOC);
	\draw[thick] (IOC) to (BIR+);
	\draw[thick] (BIR) to (BCR);
	\draw[thick] (BIR+) to (IRC);
	\draw[thick] (IRC) to (CRC);
	\draw[thick] (COC) to (CRC);
	\draw[thick] (NL) to (ELS);
	\draw[thick] (NL) to (ELR);
	\draw[thick] (ELS) to (ES);
	\draw[thick] (ELR) to (ES);
	\draw[thick] (CRC) to (CRC+ES);
	\draw[thick] (ES) to (CRC+ES);
	\draw[thick, double,double distance=0.4mm] (BCR) to (BIR+);
	\draw[thick, double,double distance=0.4mm] (BIR+) to (COC);
	\draw[thick, dashed] (IOC) to (BCR);
	\draw[thick, dashed] (COC) to (IRC);
\end{tikzpicture}
\end{center}
\caption{Hierarchy of rationality principles with respect to solvable games.
Normal lines represent proper inclusions in both the general \emph{and} 2-player case.
\emph{Double} lines represent proper inclusions in the general case---in the 2-player case there is an identity. 
\emph{Dashed} lines represent proper inclusions in the 2-player case---in the general case the two sets are not comparable.
}
\label{fig:table2}
\end{figure}

It is natural to ask whether there exists a \emph{strongest} purely rational principle. That is, a principle P which is followed by all ideally rational players and which can solve more games than any other purely rational principle. Such a principle indeed exists, since we can simply take the \emph{intersection} of \emph{all} purely rational principles to obtain a principle which is purely rational and which can solve all games that are solvable by some purely rational principle. Alternatively, we obtain the (same) strongest purely rational principle by forming a set of all \emph{protocols} corresponding to all ideally rational players (recall that rational players can be identified with the protocols that they follow).

However, even though there really exists a strongest purely rational principle, we cannot define it explicitly, unless we know all the ideally rational players  (which we have taken as a primitive notion). Of the principles presented in this paper, $\mathrm{CRC\cap ES}$ is the strongest one which could be claimed to be purely rational, but we leave it open whether this principle could be strengthened any further by intersecting it with some other purely rational principles\footnote{It is also
conceivable that $\mathrm{CRC\cap ES}$ may turn out not purely rational due to being incompatible with some other principle that is found to be purely rational. But even if this was the case, we could still argue that the \emph{informal background intuitions} behind CRC and ES are purely rational, even if their formal extensions (i.e., the corresponding sets of 
protocols) are not (cf. Section~\ref{sec: Merging purely rational principles}).}.

\section{Coordination with structural conventions \\ and structurally unsolvable games}\label{sec:conventions}\label{sec:unsolvable}



So far we have presented several principles with different levels of justification for being purely rational and 
therefore naturally applicable without any preplay communication between the players. 
However, in general, many structural
principles lack a rational justification and they may even look completely arbitrary.
Such principles can be called \emph{structural conventions}.
Even though possibly arbitrary, structural conventions could nevertheless be explicitly 
agreed upon by a group of players.
Alternatively, they could also be considered to
emerge naturally with little or no explicit negotiations. (Even highly arbitrary
conventions could conceivably be followed in different, possibly hypothetical, communities or cultures.)

%
We now discuss briefly how \WLC games can be solved if the players can 
negotiate (communicate) \emph{before} they are presented with the \WLC game they are to play.
We make the \emph{assumption} that via preplay communication, the players
can agree on which \emph{structural principles} they will use in the game.
Thus the assumption here is that the players can discuss the details of different games only up to
structural equivalence.
This assumption is natural, e.g., in scenarios where there is no obvious reason for preferring some
names of potential game choices with the expense of others. It can even be plainly impossible to
distinguish structurally equivalent choices. In these kinds of scenarios it is
natural to assume reasoning only up to structural equivalence.

Even when entirely arbitrary structural principles (i.e., structural conventions) are 
allowed, there exist \WLC games that cannot be solved.
Games that cannot be solved with \emph{any} structural principle
are called  \defstyle{structurally unsolvable}. Probably the simplest nontrivial example of such a game is 
$G(2(1\times 1))$.
Next we will characterise the class of all structurally unsolvable \WLC games.

We say that $G$ is \defstyle{structurally indeterminate} 
if all choice profiles in $W_G$ exhibit a bad symmetry (recall the definition of the principle ES).
%
%
%
%
For example,
the game $G(1\times 2 + 2\times 1 + 3\times 3 + 3\times 3)$ is
structurally indeterminate because every choice profile exhibits a bad symmetry.
%
%
On the other hand, the game 
$G(1\times 2 + 2\times 1 + 3\times 3 + 4\times 4)$ is
\emph{not} structurally indeterminate, as none of the winning choice profiles with nodes of degree $3$
or $4$ exhibit a bad symmetry.

\begin{theorem}
No structural principle can solve a structurally indeterminate game.
\end{theorem}

\begin{proof}
For the sake of contradiction, suppose that there is a structural principle P and a structurally indeterminate \WLC game $G$ such that $G\in s(\mathrm{P})$. Let $\Sigma$ be any protocol in the principle P. Since P is a structural principle, $\Sigma$ must be a structural protocol. Since $\mathrm{P'}\subseteq\mathrm{P}$ implies $s(\mathrm{P})\subseteq s(\mathrm{P'})$, the also the singleton principle $\{\Sigma\}$ solves $G$. 

Let $(u_1,\dots,u_n)\in \Sigma(G,1)\times\cdots\times\Sigma(G,n)$. Since $G$ is structurally indeterminate, $(u_1,\dots,u_n)$ must exhibit a global losing symmetry. Therefore there is a choice profile $(u_1',\dots,u_n')\in U_1\times\cdots\times U_n$ such that $(u_1',\dots,u_n')\notin W_G$. Since $\Sigma$ is a structural protocol, we must have $(u_1',\dots,u_n')\in \Sigma(G,1)\times\cdots\times\Sigma(G,n)$. Since $(u_1',\dots,u_n')\notin W_G$, we have $\Sigma(G,1)\times\cdots\times\Sigma(G,n)\not\subseteq W_G$. Therefore $\{\Sigma\}$ does not solve $G$, which is a contradiction.
\end{proof}

\medskip
The characterisation given by the theorem above is optimal in the sense that all games that are not structurally indeterminate, can be solved by some structural principle. This follows from the following even stronger result.

\begin{theorem}
There exists a protocol $\Sigma$ such that the principle $\{\Sigma\}$ solves all \WLC games that are not structurally indeterminate.
\end{theorem}

\techrep{
\begin{proof}
The idea is simply to define a protocol that chooses, on an
arbitrary input $(G,i) = ((A,C_1,\dots ,C_n,W_G),i)$,
where $G$ is not structurally indeterminate, a 
node from a tuple of $W_G$ that does not exhibit global losing symmetry.
The only difficulty is that there may be several such tuples in $G$, and these
tuples do not necessarily form a Cartesian product. We next briefly describe how to circumvent this problem. Informally, we will just
consider different linear orderings of the input structure (as well as
structures obtained from it by a player renaming) and choose the
lexicographically smallest suitable tuple.

Firstly, we use the encoding of relational structures by binary
strings given in Chapter 6 of \cite{libki}. Within this standard
encoding scheme, every string
encoding of a structure requires a linear 
ordering of the domain of the encoded structure. The linear order then defines a
lexicographic order of the tuples of the structure. The
order is used for encoding each relation $C_i$ (with $C_1$
first and $C_n$ last) and also $W_G$.
We note that a single structure will typically have several
encodings, as different linear orderings of the domain can
define different encoding strings. 
Non-isomorphic structures will never share the same encoding.

Now, given an input $(G,i) = ((A,C_1,\dots ,C_n,W_G),i)$, we do the following.
We first define the finite set $\mathcal{G}$ that contains
exactly all structures that can be obtained from $G$ by a 
player renaming, including $G$ itself. We will below refer to
the structures in $\mathcal{G}$ as \emph{renamings of} $G$.
We note that no two structures in $\mathcal{G}$ are isomorphic,
but all structures in $\mathcal{G}$ share the same domain.
Then---having defined $\mathcal{G}$---we
investigate the finite set $S$ that that contains, for
\emph{every} linear ordering $<_A$ of $A$ and 
\emph{every} structure in $\mathcal{G}$, the binary encoding of that 
structure with respect to $<_A$.
We choose the string $s\in S$ with the least binary number.
Let $G_s\in \mathcal{G}$ be a structure encoded by $s$. 
Using $s$, we choose the lexicographically smallest tuple $\vec{c}$
from the winning relation of $G_s$ that does not
exhibit global losing symmetry.\footnote{There can be
more than one linear ordering $<_A$ leading to
the string $s$, so there can be many ways to
choose the tuple $\vec{c}$. This will not be a
problem for us; we will see that it makes no
difference which tuple $\vec{c}$ we
choose. Note that there will always be an isomorphism
connecting different such tuples $\vec{c}$.}
Thus, altogether,
we obtain a renaming $G_s$ of $G$ together with the tuple $\vec{c}\in W_{G_s}$ that
does not exhibit global losing symmetry.

Let $c_j$ denote some coordinate of $\vec{c}$
such that there exists a full renaming $(\beta,\pi)$ from $G_s$ to $G$ that
associates the player role number $j$ with the player role number $i$. There may be 
several such coordinates $j$ and several renamings $(\beta,\pi)$ for $j$.
Let $D$ be the subset\footnote{It will become clear that $D$ is independent of
which tuple $\vec{c}$ we chose.} of the domain of $G$ that contains
exactly all choices $c_i$ of player $i$ in $G$ that can be obtained by 
such renamings $(\beta,\pi)$ and coordinates $j$. The
desired protocol outputs $D$ on the input $(G,i)$.
\end{proof}
}

\medskip

\lorip{For a proof, see \cite{LORI-VI-techrep}.} 

As discussed above, structural conventions may be quite artificial and arbitrary, and certainly cannot always be considered purely rational. It seems very difficult to separate, in a natural and commonly acceptable way, all purely rational principles from structural conventions.

\lorip{\vspace{-1mm}}
\section{Concluding remarks} 
\lorip{\vspace{-1mm}}

We have proposed and studied a hierarchy of principles of rational coordination that can be applied in rational players' reasoning about how to act in pure coordination scenarios. To make our study precise we have formalised such scenarios as \WLC games and have compared the scope of applicability and strength of the various principles in terms of the classes of \WLC games solvable by them. One major conclusion we draw is that 
the very questions of what is rational reasoning and what are rational choices in 
pure coordination games without preplay communication or conventions, appears to be very subtle and non-trivial. In particular, it seems very difficult to separate purely rational principles from others that can be employed in rational players' reasoning when there is no purely rational solution. On the other hand, there is a precise technical description of the class of  \WLC games that are solvable when structural conventions are enabled. 

A number of conceptual and technical issues arising from the present study remain open for further exploration. 
\begin{itemize}
\item To begin with, in this paper we have focused on scenarios where players look for choices that \emph{guarantee} winning if a suitable rational principle is followed. But it is very natural to ask how players should act in a game which seems not  solvable by any purely rational principle\footnote{We note that if players were ultimately interested \emph{only} in \emph{guaranteeing} a win in \WLC-games, even the non-losing principle NL could be questioned. This is simply because, in structurally unsolvable games (such as $G(2(1\times 1) + \overline{1\times 1})$), winning is not guaranteed by any structural principles---regardless whether players follow NL or not. This demonstrates that rational players should, of course, not only be interested in guaranteeing a win when that is possible. 
}. 
If players cannot guarantee a win, it is natural to assume that they should at least try to maximize somehow their collective chances of winning, say, by considering protocols involving some \emph{probability distribution} over their choices. 

\item Another natural extension of the present framework is to consider non-structural principles based on limited preplay communication and use of various types of non-structural conventions based on some additional features of the game representation, e.g., partial priority orders of players, colours of choices, etc. 
For a general exploration of \WLC games with enriched game representation see \cite{DBLP:conf/eumas/GorankoKR17}.

\item Studying pure dis-coordination games and games involving combinations of coordination and dis-coordination types of players are major potential directions for further work. 
\end{itemize}

Finally, to provide contrast with the theoretical work presented in this paper, in the future we plan to run empirical experiments on people's behaviour \WLC games.

\subsection*{Acknowledgements}

The work of Valentin Goranko was partly supported by a research grant 2015-04388 of the Swedish Research Council. 
The work of Antti Kuusisto was supported by the ERC grant 647289 ``CODA"
and a part of the work also by the Academy of Finland 
grants  438 874 and 209 365.
The work of Raine R\"onnholm was partly supported by Jenny and Antti Wihuri Foundation.
We thank the reviewers of this paper, as well as the reviewers of \cite{sr2017}, \cite{lori2017}, and \cite{DBLP:conf/eumas/GorankoKR17}, for helpful remarks and references.

\bibliographystyle{plain}
\bibliography{GT-Bibliography}

\end{document}